\documentclass[10pt,journal]{IEEEtran}

\usepackage{c_auxiliary/Myshorts}
\usepackage{c_auxiliary/GalShorts}
\usepackage{c_auxiliary/gal_packages}
\usepackage{c_auxiliary/acronyms}
\usepackage{f_math/math}

\usepackage{booktabs} 

\usepackage{subcaption}

\usepackage{catchfilebetweentags}

\usepackage{tikz,forest}
\usetikzlibrary{shapes,arrows}

\newtheorem{defin}{Definition}
\newtheorem{exm}{Example}
\newtheorem{thm}{Theorem}

\newtheorem{rem}{Remark}
\newtheorem{claim}{Claim}

\newtheorem{lem}[thm]{Lemma}

\usepackage{titlesec}

\definecolor{rev}{HTML}{1679AB}
\definecolor{reg}{HTML}{000000}

\usepackage{graphicx}
\usepackage{float}

\begin{document}

\title{Efficient Recovery of Sparse Graph Signals from Graph Filter Outputs }

\author{Gal Morgenstern~\IEEEmembership{Student Member,~IEEE},
	and~Tirza~Routtenberg,~\IEEEmembership{Senior Member,~IEEE}  
\thanks{G. Morgenstern and T. Routtenberg are with the School of ECE, Ben-Gurion University of the Negev, Beer Sheva, Israel (e-mail: \{galmo@post.;tirzar\}@bgu.ac.il). 
This work is partially supported by the Israel Science Foundation (grant No. 1148/22), the Jabotinsky Scholarship from the Israel Ministry of Technology and Science, and the Israel Ministry of National Infrastructure and Energy.
Parts of this work (the GM-GIC algorithm) were presented at the IEEE  Computational Advances
in Multi-Sensor Adaptive Processing (CAMSAP) Workshop 2023 \cite{morgenstern2023sparse}. 
} }

\maketitle

\begin{abstract} \label{sec; abstract}
This paper investigates the recovery of a node-domain sparse graph signal from the output of a graph filter.
This problem, which is often referred to as the identification of the source of a diffused sparse graph signal, is seminal in the field of graph signal processing (GSP).
Sparse graph signals can be used in the modeling of a variety of real-world applications in networks, such as social, biological, and power systems, and enable various GSP tasks, such as graph signal reconstruction, 
blind deconvolution, and sampling. 
In this paper, we assume double sparsity of both the graph signal and the graph topology, as well as a low-order graph filter. 
We propose three algorithms
to reconstruct the support set of the input sparse graph signal from the graph filter output samples, leveraging these assumptions and the generalized information criterion (GIC). 
First, we describe the graph multiple GIC (GM-GIC) method, which is based on partitioning the dictionary elements (graph filter matrix columns) that capture information on the signal into smaller subsets.
Then, the local GICs are computed for each subset and aggregated to make a global decision. 
Second, inspired by the well-known branch and bound (BNB) approach, we develop the graph-based branch and bound GIC (graph-BNB-GIC), and
incorporate a new tractable heuristic bound tailored to the graph and graph filter characteristics.
 In addition, we propose the graph-based first order correction (GFOC) method, 
which improves existing sparse recovery methods by iteratively examining 
potential improvements to the GIC cost function by 
replacing elements from the estimated support set with elements from their
one-hop neighborhood.
Simulations on  stochastic block model (SBM) graphs demonstrate that the proposed sparse recovery methods outperform 
existing techniques
 in terms of support set recovery and mean-squared-error (MSE), without significant computational overhead. 
In addition, we investigate the application of our graph-based sparse recovery methods in blind deconvolution scenarios where the graph filter is unknown. Simulations using real-world data from brain networks and pandemic diffusion analysis further demonstrate the superiority of our approach compared to graph blind deconvolution techniques.

\end{abstract}



\section{Introduction}
\label{sec:introduction}

\Ac{gsp} theory extends classical time-varying signals to irregular domains represented by graphs \cite{ortega2018graph,shuman2013emerging,sandryhaila2013discrete}. 
The graph under consideration may represent a physical network, such as electrical or sensor networks, or a virtual network, such as social networks, or it may encode statistical dependencies among the signal values. 
Recent developments in \ac{gsp}  encompass a wide array of processing tools,
including graph spectral analysis \cite{sandryhaila2014discrete,marques2017stationary}, anomaly detection \cite{drayer2019detection}, sampling and signal recovery \cite{tanaka2020sampling,marques2015sampling,chen2015signal,di2018adaptive,sagi2022gsp,kroizer2022bayesian,amar2023widely}, \color{black}
verification of graph smoothness \cite{dabush2023verifying,venkitaraman2019predicting}, and graph filter design \cite{ortega2018graph,shuman2013emerging,isufi2016autoregressive,coutino2019advances}.
However, limited attention has been paid to developing \ac{gsp}-based techniques for the sparse recovery of node-domain {\em{sparse graph signals}}.

Signal modeling based on sparse representations is used in numerous signal and image processing applications
\cite{Chen_Donoho_2001,Tropp_2004,donoho2005stable,elad2010sparse}. Sparse recovery methods can broadly be categorized into three main groups: convex relaxations, often using methods such as $\ell_1$-norm minimization or basis pursuit \cite{Chen_Donoho_2001,tibshirani1996regression,efron2004least,candes2007dantzig, donoho2009message}, non-convex optimization approaches \cite{wipf2004sparse,gorodnitsky1997sparse, chartrand2008iteratively}, and greedy algorithms \cite{mallat1993matching,Tropp_2004,cai2011orthogonal, dai2009subspace, blumensath2008gradient }. 
In the context of estimation theory, Bayesian inference approaches estimate the sparse signal by incorporating prior information, while other methods use the structure of the dictionary matrix. 
In network science, compressed sensing and group testing have been discussed in \cite{haupt2008compressed,xu2011compressive,cheraghchi2012graph},
and diffusion processes that originate from a sparse signal have been discussed in \cite{pinto2012locating,sefer2016diffusion,zhang2016towards,pena2016source,sridhar2023quickest}.  \label{pp; sampling works} Several \ac{gsp} studies have investigated the sampling and reconstruction of graph signals by assuming the signal belongs to a vector subspace. 
For example, the works  \cite{wang2015local,anis2014towards,chen2015discrete,tsitsvero2016signals} assume graph bandlimited signals, while \cite{mashhadi2017interpolation} considers sparsity in the spectral domain. Similarly, \cite{romero2016kernel} models the graph signals as smooth in the node domain or decaying in the spectral domain. 

Recently there has been growing interest in setups where a node-domain sparse 
graph signal $\signal$ is filtered by a graph filter $\Hmat$. 
In this case, the output signal $\measurements$ is referred to as a diffused sparse graph signal \cite{ramirez2021graph}.  
Modeling diffusion processes over graphs by graph filters is valuable due to the ability to capture local interactions among nodes \cite{segarra2016blind,segarra2017optimal,segarra2015distributed,mei2015signal}.
Diffused sparse graph signals can be used for modeling many real-world network scenarios.
For instance, in social networks, these signals can be used for opinion forming \cite{rainer2002opinion} and finding the origin of a rumor \cite{shah2011rumors}.
In biological networks, they enable the analysis of the spread of a disease \cite{newman2002spread},
and in computer networks, they can be used to locate malware \cite{shah2010detecting}. 
In power systems, they may be used for locating anomalies in power signals
\cite{morgenstern2020structuralconstrained,morgenstern2023protection}. 
In the field of \ac{gsp}, diffused sparse graph signals have been studied 
in several problems, such as bandlimited graph signal reconstruction \cite{segarra2016reconstruction}. 
Blind deconvolution, in which both the graph filter and the sparse input are estimated, has been discussed in \cite{ramirez2021graph,segarra2016blind,zhu2020estimating,ye2018blind,iglesias2018demixing}.
System identification, considered in \cite{segarra2017optimal,ramirez2021graph}, refers to the case where the sparse signal is known, but the graph filter coefficients are estimated. 
The sampling of diffused sparse graph signals is discussed in \cite{rey2019sampling}. 
For an overview of blind deconvolution and signal reconstruction methods for diffused sparse graph signals, we refer the readers to Fig. 1 in \cite{ramirez2021graph}. 
However, in most of the applications discussed above, the sparse recovery does not exploit the \ac{gsp}-based modeling 
of the {\em{sparse graph signals}}. 
Instead, these methods typically rely on standard sparse approximations, such as $\ell_1$ relaxation, or employ tailored solutions that use specific properties of the application at hand.

In this paper, we establish efficient sparse recovery methods for 
graph signals using
the \ac{gic} \cite{stoica2004model} \color{black} as the cost function.
We assume double sparsity of both the graph and the graph signal and that the network diffusion 
process is modeled by a known low-degree graph filter. 
A direct solution for the \ac{gic}-based problem necessitates an exhaustive search, whose  
complexity grows exponentially with the size of the graph,
rendering it impractical. 
We propose three graph-based sparse recovery methods that harvest the double sparsity imposed by the graph and the graph signal. 
The first method, the \ac{gmgic} method,
identifies and partitions a set of suspected nodes and then applies a local \ac{gic} on each partition, which significantly reduces the computational overhead of the \ac{gic}-based approach.
The second method, the \ac{gbnb} method, further mitigates the computational demands by 
iteratively searching over the candidate support sets of the sparse graph signals (the same sets as in the \ac{gmgic} method).
This method is developed based on the \ac{bnb} method \cite{boyd2007branch, lawler1966branch} with a heuristic graph-based upper bound that exploits the double sparsity and the graph properties. 
The third method, the \ac{gfoc} method, builds upon existing sparse recovery techniques and improves them in terms of the GIC cost function by 
replacing elements from the estimated support set with elements from their
one-hop neighborhood.
Our simulations on \ac{sbm} graphs demonstrate that the proposed methods outperform state-of-the-art techniques: \ac{omp}, Lasso, and \ac{bnb} with $\ell_1$ regularization methods in terms of the support set recovery accuracy  and \ac{mse},\color{black} 
without a significant computation overhead. 
 In addition, to address scenarios where the graph filter is unknown, we introduce a new approach that incorporates a preliminary blind deconvolution stage to estimate the graph filter. 
Our simulations demonstrate that the proposed graph-based recovery methods can significantly improve the performance of existing graph blind deconvolution techniques in real-world applications, including: (1)
brain networks obtained from real data 
 \cite{hagmann2008mapping}; and (2) 
 the identification of the source of a disease using data from the $1854$ cholera outbreak  \cite{wilson2012cholera,snow1854report}. 

\label{pp; introduction blind 1}

The paper is organized as follows. 
Section \ref{sec; problem} presents the necessary background and outlines the problem formulation. 
In Section \ref{sec; analysis}, we analyze the dictionary matrix for the proposed graphical setting. 
We describe both the \ac{gmgic} and \ac{gbnb} methods in Section \ref{Sec; theory sparse recovery}, followed by the presentation of the \ac{gfoc} method in Section \ref{sec; sparse; gfoc}.
In Section \ref{sec; simulations}, we present the simulation results, and the conclusions are provided in Section \ref{sec; conclusions}.

{\em{Notations:}} 
In the rest of this paper, vectors and matrices are denoted by boldface lowercase and uppercase letters, respectively. 
The $(\none,\ntwo)$th element of the matrix $\filter$ is denoted by $H_{\none,\ntwo}$. 
The vector ${\signal}_{\signalsupport}$ is a subvector of $\signal$ with the elements indexed by the set $\signalsupport$, 
and $\filter_{\signalsupport}$ is the submatrix of $\filter$ consisting of the {\em{columns}} indexed by $\signalsupport$. In particular, $\filter_{\none}$ denotes the $\none$th column of $\filter$.
The projection matrix into the subspace $\text{col}(\filter_{\signalsupport})$
is given by \be
\label{proj}
\proj{\signalsupport}=\projection{\submatcol{\filter}{\signalsupport}}. 
\ee
The operators $\norm{\cdot}$, $\norm{\cdot}_1$, $\norm{\cdot}_0$, and $\card{\cdot}$ denote the Euclidean norm, 
$\ell_1$ norm, zero semi-norm, and set cardinality. 
The identity matrix is denoted by $\Imat$.


\section{Problem Formulation} \label{sec; problem}

In this section, we formulate the problem of recovering a sparse graph signal from the outputs of a graph filter. 
First, in Subsection \ref{sec; graph}, we present relevant \ac{gsp} background.
Then, in Subsection \ref{sec; assumptions}, we outline the graphical setting 
considered.
Finally, in Subsection \ref{sec; sparse}, we present the recovery problem addressed in this paper.

\subsection{Background: Graphs and Graph Filters } \label{sec; graph}

Consider an undirected weighted graph  $\graph$, which consists of a set of nodes (vertices) $\nodes$, a set of edges $\edges$, 
and a set of weights $\{\edgeweight_{e}\}_{e\in\edges}$. 
A graph signal is a mapping from the graph nodes into a $\card{\nodes}\times 1$ real vector:
\be \label{eq; graph signal}
\signal: \nodes\rightarrow \mathbb{R}^{\card{\nodes}}.
\ee
A \ac{gso}, $\dimsquare{\gso}{\card{\nodes}}$, is a linear operator applied on graph signals that satisfies
\be\label{eq; gso}
  \gsoentry_{\none,\ntwo}=0, ~ ~~\text{if}~\geo{\none}{\ntwo}>1,~~~\forall \none,\ntwo\in\nodes,
\ee
where the geodesic distance between nodes $\none$ and $\ntwo$, $\geo{\none}{\ntwo}$, is the number of edges on the shortest path between
these nodes. 
According to \eqref{eq; gso},
the non-zero elements of $\entrymat{\gso}{\none}{\ntwo}$ are the diagonal elements (i.e. $\none=\ntwo$) or the elements that are associated with an edge (i.e. $\edge{\none}{\ntwo}\in\edges$). 
As a result, by applying a \ac{gso} on a graph signal $\signal$, each element of the shifted signal, $\gso\signal$, is a linear combination of the signal elements in its one-hop neighborhood:
\be \label{eq; gso linear}
[\gso\signal]_{\none}= \sum_{\ntwo\in
\nei{1}{\none}} \gsoentry_{\none,\ntwo} \signalentry_{\ntwo}, \quad \none=1,\ldots,\card{\nodes}, 
\ee
where 
\be
\label{nei_def}
\nei{\Delta}{\none}=\{\ntwo\in\nodes:~0\le \geo{\none}{\ntwo}\le \Delta\},~\forall \none\in\nodes.
\ee
Thus, the transformation by the \ac{gso} can be computed locally at each
node by aggregating the values of the input signal within the one-hop neighborhood of each node.

A shift-invariant graph filter, $\dimsquare{\filter}{\card{\nodes}}$, is defined as a  
polynomial of the \ac{gso}\cite{shuman2013emerging,sandryhaila2013discrete}:
\be \label{eq; graph filter}
\filter=h_0\Imat+h_1\gso+\ldots +h_{\filterdeg}\gso^{\filterdeg},
\ee
where $h_0,\ldots,h_{\filterdeg}$ are the coefficients of the graph filter and $1\leq  \filterdeg\leq \card{\nodes}-1 \color{black} $  is the graph filter order. 
By substituting \eqref{eq; gso} in \eqref{eq; graph filter} it can be verified that the $(\none,\ntwo)$th element of $\filter$
satisfies
\be \label{eq; F prop}
\filterentry_{\none,\ntwo}=\sum_{i=0}^\filterdeg h_i [\gso^i]_{\none,\ntwo}=0,~\text{if}~\geo{\none}{\ntwo}>\filterdeg,~\forall \none,\ntwo\in\nodes.
\ee
Consequently, each element of the filtered signal, $\filter\signal$, is a linear 
combination of the signal elements in its  $\filterdeg$-hop neighborhood.
Hence, the filtered signal satisfies 
\be \label{eq; graph filter linear}
[\filter\signal]_{\none}= \sum_{\ntwo\in\nei{\filterdeg}{\none}} \filterentry_{\none,\ntwo} \signalentry_{\ntwo}, 
\ee
where $\nei{\filterdeg}{\none}$ is defined in \eqref{nei_def}.
Thus, the filtered signal can be computed locally at each
node by aggregating the values of the input signal within the $\filterdeg$-hop neighborhood of each of the nodes.
We conclude with the following lemma.
\begin{lem} \label{lem; F prop}
    Let, $\ntwo$ and $\none$ be two nodes with a geodesic distance larger than $2\Psi$, i.e. $\geo{\none}{\ntwo}>2\filterdeg$. 
    Then, their associated graph filter columns $\filter_{\ntwo}$ and $\filter_{\none}$ are orthogonal, i.e. 
\be \label{eq; F prop 2}
\filter_{\ntwo}^T\filter_{\none} = 0.
\ee
\end{lem}
\begin{proof}
First, we observe that 
\be
\filter_{\ntwo}^T\filter_{\none} = \sum_{\iterone\in\nodes} H_{\ntwo,\iterone}H_{\none,\iterone}.
\ee
Additionally, the assumption that $\geo{\none}{\ntwo}>2\filterdeg$ implies 
that there is no $\iterone\in \nodes$ such that both $\geo{\ntwo}{\iterone}\le\filterdeg$ and  $\geo{\none}{\iterone}\le\filterdeg$.
Hence, from \eqref{eq; F prop}, $\forall\iterone\in\nodes$, at least one of the elements $H_{\ntwo,\iterone}$ and $H_{\none,\iterone}$ is zero. 
Therefore, we obtain \eqref{eq; F prop 2}. 
\end{proof}

\subsection{Assumptions on the Graphical Setting} \label{sec; assumptions}

In this paper, 
we consider the following assumptions on the underlying graph and the graph filter.
\begin{enumerate}[label=\textbf{A.\arabic*}, leftmargin=0.8cm,topsep=2.5pt]
\item \label{A; graph degree} The maximum nodal degree (the maximal number of edges connected to one node) of the graph is considered small, i.e. $d_{max}\ll\card{\nodes}$.
\item \label{A; filter degree} The graph filter is a polynomial of the \ac{gso}
as defined in \eqref{eq; graph filter}. Additionally, the graph filter degree, $\filterdeg$, is small \ac{wrt} the number of nodes, i.e. 
$ \filterdeg\ll\card{\nodes}$.
\end{enumerate}

The meaning of these assumptions is as follows. 
Assumption \ref{A; graph degree} enforces a sparse connectivity pattern on the underlying graph, i.e. {\em{sparse graph}},  where the number of edges is limited by 
$ d_{max}\card{\nodes}$, 
which is far fewer than the maximum possible number of edges ($(1/2)\card{\nodes}(\card{\nodes}-1)$).
Furthermore, this sparse connectivity pattern is enforced uniformly in the sense that 
any subgraph of the graph is also sparse, 
where the restriction in Assumption \ref{A; graph degree}
also applies to the subgraph.   
Assumption \ref{A; filter degree} restricts the number of elements in each $\filterdeg$-hop neighborhood
by imposing a small value for $\filterdeg$. 
The structure of the polynomial graph filter shown in \eqref{eq; graph filter} also
implies that, in general,
the values of the filtered signal, $\filter\signal$, in \eqref{eq; graph filter linear}, 
are more significant for indices associated with nodes closer to the nodes in $\signalsupport_t$.
Here, $\signalsupport_t$ refers to the support set of the sparse signal. 
Therefore, combining Assumptions \ref{A; graph degree}-\ref{A; filter degree} has a direct impact on the filtered signal outputs 
in \eqref{eq; graph filter linear}, which are shown to be linear combinations of 
signal elements in their $\filterdeg$-hop neighborhoods.
Consequently, only a portion of the 
filtered signal outputs hold substantial information on the input signal. 

The graph filter in \eqref{eq; graph filter} under Assumptions \ref{A; graph degree}-\ref{A; filter degree},
 in addition to the sparsity restriction in \eqref{eq; sparsity restriction}, 
can be used to model a variety of \ac{gsp} tasks, including 
 signal reconstruction, blind deconvolution, sampling, and system identification \cite{ramirez2021graph,segarra2016blind,zhu2020estimating,segarra2017optimal,rey2019sampling}.
 Additionally, this framework can be used for modeling real-life applications 
 in social networks \cite{rainer2002opinion,shah2011rumors}, 
 biological networks \cite{newman2002spread},  brain networks \cite{hagmann2008mapping},
 epidemiology, 
 computer networks  \cite{shah2010detecting}, 
 wireless sensor networks \cite{sandryhaila2014discrete}, and power systems \cite{morgenstern2020structuralconstrained,morgenstern2023protection}. 
It is worth noting that in the context of power systems, the average nodal degree is generally very low (usually between 2 to 5) and does not increase with the network size \cite{wang2010generating}. This property is also common in other infrastructure systems, such as water and transportation systems.

\subsection{Recovery of Sparse Graph Signals} \label{sec; sparse}
In this paper, we aim to recover the source of a diffused sparse graph signal \cite{ramirez2021graph} from noisy observations under Assumptions \ref{A; graph degree}-\ref{A; filter degree}.
That is, we are interested in recovering the graph signal $\signal$, defined in \eqref{eq; graph signal}, from the observation vector $\measurements$, where the following measurement model is assumed:
\be \label{eq; model}
\measurements=\filter\signal+\noise.
\ee
In this model, $\dimsquare{\filter}{\card{\nodes}}$ is the graph filter, defined in \eqref{eq; graph filter}, which is a sparse dictionary matrix that 
satisfies Assumptions \ref{A; graph degree}-\ref{A; filter degree}.
The noise is modeled by a $\card{\nodes}\times 1$ zero-mean Gaussian vector, $\noise$, with a known 
covariance matrix, $\noisecov$, i.e. $\noise\sim\normal{\zerovec}{\noisecov}$.
Additionally, the unknown graph signal is assumed to be sparse:
\be \label{eq; sparsity restriction}
\norm{\signal}_0\leq \sparsityterm,
\ee
where $\norm{\cdot}_0$ is the $\ell_0$ semi-norm and the sparsity parameter $\sparsityterm$ indicates the sparsity level.
Hence, there exists a set of nodes, $\signalsupport_t$, referred to as the signal support set  
such that 
\be \label{eq; support set}
\entryvec{\signal}{\none}=0,~\forall \none \in \{\nodes\setminus\signalsupport_t\}.
\ee
That is, only the elements in the support set $\signalsupport_t$ may correspond to non-zero elements of  $\signal$.

Based on \eqref{eq; model}-\eqref{eq; support set}, the sparse recovery problem can be formulated 
as (Eq. (1) in \cite{Tropp_2004})
\be \label{eq; SR 2}
(\hat{\signal},\hat{\signalsupport})=\arg\underset{\card{\signalsupport} \le \sparsityterm} {\min}
\underset{~\dimvec{\signal_{\signalsupport}}{\card{\signalsupport}}}{\min} ~\norm{ \measurements-
\filter_\signalsupport \signal_{\signalsupport}} ^2.
\ee
The inner minimization problem in \eqref{eq; SR 2} is 
a \ac{ls} problem with the solution (see, e.g. p. 225 in \cite{kay1993fundamentals}):
\be \label{eq; WLS}
\hat{\signal}_{\signalsupport}=\arg\underset{\dimvec{\signal_{\signalsupport}}{\card{\signalsupport}}}{\min} ~\norm{ \measurements-
\filter_\signalsupport \signal_{\signalsupport}}
^2=(\filter_{\signalsupport}^T\filter_{\signalsupport})^{-1}\filter^T_{\signalsupport}\measurements.
\ee
By substituting \eqref{eq; WLS} in \eqref{eq; SR 2} 
we obtain the following support set recovery problem:
\be \label{eq; sup}
\hat{\signalsupport}=\arg\underset{\card{\signalsupport} \le \sparsityterm}{\min}\norm{ \measurements-\proj{\signalsupport}\measurements}^2 =\arg\underset{\card{\signalsupport} \le \sparsityterm}{\max}\norm{\proj{\signalsupport}\measurements}^2,
\ee
where $\proj{\signalsupport}$ is the projection matrix onto the suspace $\text{col}(\filter_{\signalsupport})$, defined in \eqref{proj}.
The last equality in \eqref{eq; sup} is obtained by using the projection matrix property (see Theorem 2.22 in \cite{yanai2011projection}): $$\norm{\measurements-\proj{\signalsupport}\measurements}^2= \norm{\measurements}^2-\norm{\proj{\signalsupport}\measurements}^2,$$ removing the constant $\norm{\measurements}^2$, and rearranging the problem. 

The support set recovery problem presented in \eqref{eq; sup} applies to any sparse signal $\signal$ and dictionary matrix $\filter$, and is not restricted to graph-based settings. This formulation is valid as long as $\filter_{\signalsupport}^{T}\filter_{\signalsupport}^{~}$ is a nonsingular matrix. \label{pp; nonsingular} 
Therefore, we assume that for any support set $\Omega$ that satisfies 
$\card{\signalsupport}\le s$, i.e. support that is relevant to our algorithms, 
the matrix $\filter_{\signalsupport}$ has a full column rank.  \color{black}
This problem can be viewed as a 
multiple-hypothesis testing problem, 
where under each hypothesis, the sparse signal has a different support set \cite{babu2023multiple}, which belongs to the following set of candidate sets: 
\be \label{eq; candidate set}
\{\signalsupport\in\nodes:~\lvert \signalsupport\rvert\le s\}.
\ee

A direct solution of the support set recovery in \eqref{eq; sup} tends to overfit by selecting support sets with larger cardinality than the true support set. 
To mitigate this problem, we use the well-known \ac{gic} model selection method,
 which is a generalization of various criteria, such as the \ac{aic}
 and the \ac{mdl}.
 For the considered problem, the \ac{gic} is given by \cite{stoica2004model}
\be \label{eq; gic}
\hat{\signalsupport}=\arg \underset{\card{\signalsupport} \le \sparsityterm}{\max}\norm{\proj{\signalsupport}\measurements}^2-\penaltygic{\signalsupport} .
\ee
Compared to \eqref{eq; sup}, the \ac{gic} in \eqref{eq; gic} includes an additional penalty function component, $\penaltygic{\signalsupport} $, which depends on the support set cardinality. 
This penalty term mitigates the risk of overfitting; a larger penalty will lead to less complex models with a smaller support set, while a smaller penalty will result in a larger support set \cite{stoica2004model}.
The penalty function can be tuned based on historical data, through trial and error,  
or by using model prior information.
After solving \eqref{eq; gic},
the sparse signal recovery solution can be obtained by substituting $\hat{\signalsupport}$ in \eqref{eq; WLS}.

In practice, an exact solution for \eqref{eq; gic} requires an exhaustive search over all $\sum_{\iterone=1}^{\sparsityterm}\binom{\card{\nodes}}{\iterone}$ optional support sets, which is infeasible.
To overcome the unfeasibility of the combinatorial search, we propose using the graphical properties imposed by Assumptions \ref{A; graph degree}-\ref{A; filter degree}, and the {\em{double sparsity}} of the graph and the graph signal.
The analysis for this setting is presented in Section \ref{sec; analysis}, and the 
proposed low-complexity support set recovery methods are developed in Section \ref{Sec; theory sparse recovery} and Section \ref{sec; sparse; gfoc}. 

\section{Dictionary Matrix Analysis} \label{sec; analysis}
In this section, we analyze the dictionary matrix atoms (i.e. graph filter columns) 
based on the underlying graphical structure and Assumptions \ref{A; graph degree}-\ref{A; filter degree}. 
This analysis forms the basis for the proposed sparse recovery methods in Sections
\ref{Sec; theory sparse recovery} and \ref{sec; sparse; gfoc}. 

Taking into account that $x_{\none}=0$ for any $\none\notin \signalsupport_t$, $\none\in\nodes$ (see \eqref{eq; support set}), the filtered signal satisfies
\be \label{eq; filtered signal}
\filter\xvec=\filter_{\signalsupport_t}\signal_{\signalsupport_t}=\sum\limits_{\none\in\signalsupport_t} \filter_{\none}\signalentry_{\none}, 
\ee
where $\filter_{\none}$ is the $\none$th column of $\filter$.
Hence, the filtered signal, $\filter\signal$, can be viewed as a linear combination of dictionary atoms (graph filter columns). 
 Therefore, projecting the filtered signal into the single-column space 
$\text{col}(\filter_{\ntwo})$ results in
\be \label{eq; projected filtered signal}
\proj{\ntwo}\filter\xvec=\sum_{\none\in\signalsupport_t} \proj{\ntwo}\filter_{\none}\signalentry_{\none}, 
\ee
where $\proj{\ntwo}$ (defined in \eqref{proj}) is a linear operator. 
Thus, it can be verified that 
\be \label{eq; projection signal condition}
\proj{\ntwo}\filter\xvec=\zerovec~\text{if}~\proj{\ntwo}\filter_{\none}=\zerovec,~\forall \none\in\signalsupport_t.
\ee

In the following lemma, we show that $\proj{\ntwo}\filter_{\none}=0$ if 
the geodesic distance between $\none$ and $\ntwo$ is larger than $2\filterdeg$, $\geo{\none}{\ntwo}>2\filterdeg$.  

\begin{lem} \label{lem; projection single}
 If two dictionary atoms correspond to nodes $\none$ and $\ntwo$ with a geodesic distance
larger than $2\filterdeg$, i.e. $\geo{\none}{\ntwo}>2\filterdeg$, then 
\be \label{eq; 2nd 0}
\proj{\ntwo}\filter_{\none}
=\projection{\filter_{\ntwo}}\filter_{\none}  = \zerovec,
\ee
where $\proj{\ntwo}$ is
the projection matrix onto the single-column space $\text{col}(\filter_{\ntwo})$.

\end{lem}
\begin{IEEEproof}
The first equality in \eqref{eq; 2nd 0} is obtained by substituting the definition of the projection matrix $\proj{\ntwo}$.
The second equality is obtained by substituting 
 \eqref{eq; F prop 2} from Lemma \ref{lem; F prop}, since it is given that $\ntwo$ and $\none$ satisfy $\geo{\ntwo}{\none}>2\filterdeg$.
\end{IEEEproof}

Leveraging the results in \eqref{eq; filtered signal}-\eqref{eq; 2nd 0}, 
we present the following definition and theorem.

\begin{defin} \label{def; neighverhood}
The $\filterdeg$-order neighborhood of the set $\signalsupport_t$ is 
\beqna \label{eq; neighborhood}
\nei{2\filterdeg}{\signalsupport_t}\define \bigcup_{\none\in\signalsupport_t} \nei{2\filterdeg}{\none} \hspace{4.35cm}\nonumber\\
=\{\none\in\nodes:~\exists\ntwo\in\signalsupport_t:~ 0<\geo{\none}{\ntwo}\le 2\filterdeg\},
\eeqna
where $\nei{2\filterdeg}{\none}$ is defined 
in \eqref{nei_def}.  
\end{defin}

\begin{thm} \label{thm;outside}
 If \( \ntwo \notin \nei{2\filterdeg}{\none} \), then the following hold:
\begin{subequations}
\begin{align}
&\proj{\ntwo}\filter\signal = \zerovec, \label{eq; filtered 0} \\
&\proj{\ntwo}\measurements = \proj{\ntwo}\filter\signal + \proj{\ntwo}\noise = \proj{\ntwo}\noise, \label{14}
\end{align}
\end{subequations}
where $ \proj{\ntwo}$ is defined in \eqref{proj} and $\measurements$ is the measurement model defined in \eqref{eq; model}. 
\end{thm}
\begin{IEEEproof}
Since \( \ntwo \notin \nei{2\filterdeg}{\none} \), then Lemma \ref{lem; projection single} implies  that $\proj{\ntwo}\filter\signal$ can be expressed using \eqref{eq; projected filtered signal}, where each of the components in the summation satisfies \(\proj{\ntwo}\filter_{\none}x_{\none}
 = \zerovec\). Thus, we obtain \eqref{eq; filtered 0}. 
 Additionally, by left-multiplying the measurement model in \eqref{eq; model} by \(\proj{\ntwo}\) and then substituting \eqref{eq; filtered 0} in the result, we obtain \eqref{14}.
\end{IEEEproof}

It can be seen that  Theorem \ref{thm;outside} implies that 
the dictionary atoms 
indexed by nodes outside the set $\nei{2\filterdeg}{\signalsupport_t}$ in \eqref{eq; neighborhood}
contain only the contribution of the 
noise. 
That is, the dictionary matrix elements  
that contain information on the sparse graph signal are only those included in the 
$2\filterdeg$ neighborhood of $\signalsupport_t$,  
$\mathcal{N}_{2\filterdeg}[\signalsupport_t]$.
 Thus, in practice, the true support of the signal is included in the set 
\be \label{eq; reduced feasible set}
\{\signalsupport\in\mathcal{N}_{2\filterdeg}[\signalsupport_t]:~\lvert \signalsupport\rvert\le s\}. 
\ee 
Consequently, we can significantly reduce the search set in \eqref{eq; gic}
by using \eqref{eq; reduced feasible set} instead of \eqref{eq; candidate set}. 
This point is formalized in 
the following theorem, which establishes a bound on the cardinality of the set $\nei{2\filterdeg}{\signalsupport_t}$ from \eqref{eq; neighborhood}
as a function of the graph filter degree $\filterdeg$,
the graph structure, and the sparsity level $s$. 
This theorem provides constraints on the size of
a reduced set of dictionary atoms that contain signal information.
In particular, it offers two key advantages: 
first, it reduces the number of candidate support sets used to construct the feasible set in \eqref{eq; gic}; second,   in some cases, 
the node set, $\nodes$,  could be divided into smaller subsets, enabling parallel processing.

\begin{thm}\label{thm; graph degree}
The cardinality of the set $\mathcal{N}_{2\Psi}(\signalsupport)$ is bounded by 
\be \label{eq; bound degree} 
\card{\nei{2\filterdeg}{\signalsupport_t}} \le s\sum_{\iterone=0}^{2\filterdeg} d_{\max}^\iterone,
\ee
where $d_{\max}$ is the maximum degree of the graph.
\end{thm}
\begin{proof}
Since $d_{\max}$ is the maximal degree of the graph, we obtain that each node in the graph has up to $d_{\max}$ neighboring nodes, and each of those neighboring nodes can, in turn, have up to $d_{\max}$ neighbors, and so on. 
As a result, the size of the $2\filterdeg$ neighborhood of a single node $\none$  satisfies \be
\label{k_bound}
\card{\nei{2\filterdeg}{\none}}\le\sum_{\iterone=0}^{2\filterdeg} d_{\max}^\iterone. 
\ee
By summing all the $k$th bounds in the form of \eqref{k_bound} for $k\in \signalsupport_t$,
we obtain
\beqna \label{eq; bound degree proof} 
\card{\nei{2\filterdeg}{\signalsupport_t}} \le \sum_{\none=1}^{\card{\signalsupport_t}} 
\card{\nei{2\filterdeg}{\none}}  \le \sum_{\none=1}^{\card{\signalsupport_t}} \sum_{\iterone=0}^{2\filterdeg} d_{\max}^\iterone.
\eeqna
Therefore, since $\card{\signalsupport_t} \le \sparsityterm$ under the sparsity assumption in \eqref{eq; sparsity restriction}, we conclude that the $2\filterdeg$-neighborhood
of $\signalsupport$ is bounded by the \ac{rhs} of \eqref{eq; bound degree}. 
\end{proof}

 The bound in \eqref{eq; bound degree} depends on the sparsity level, $\sparsityterm$, 
the maximal degree of the graph, $d_{\max}$, and the graph filter degree, $\filterdeg$,
which are all considered to be small based on \eqref{eq; sparsity restriction}, Assumption \ref{A; graph degree}, and Assumption \ref{A; filter degree}, respectively. 
In contrast, the bound is independent of the size of the graph, i.e. the number of nodes, $\card{\nodes}$. 
Thus, for large networks it can be assumed that the bound is significantly lower than the actual number of nodes:
\be \label{eq; bound super small}
\sparsityterm\sum_{\iterone=0}^{2\filterdeg} d_{\max}^\iterone\ll \card{\nodes}. 
\ee
Consequently, by substituting \eqref{eq; bound super small} in \eqref{eq; bound degree} we obtain
\be
\label{new_bound}
\card{\nei{2\filterdeg}{\signalsupport_t}}\ll \card{\nodes}. 
\ee
The significance of the result in \eqref{new_bound} lies in the observation that the number of dictionary atoms containing information about the sparse signal is much smaller than the total number of dictionary atoms that are required for searching for the optimal \ac{gic} in \eqref{eq; gic}. As a result, any sparse recovery method that employs a search algorithm (and not just the \ac{gic} in \eqref{eq; gic}) can significantly reduce its computational complexity by focusing on the reduced set of dictionary atoms associated with the nodes in $\nei{2\filterdeg}{\signalsupport_t}$ instead of on the entire set associated with all nodes in $\nodes$.

Moreover, it should be noted that the bound in \eqref{eq; bound degree} is not tight. 
This is because we may have overlaps between the $\filterdeg$-order neighborhoods of the nodes 
in $\signalsupport_t$ 
and overlaps within the $\filterdeg$-order neighborhoods of individual nodes (see first and second inequalities in \eqref{eq; bound degree proof}, respectively). 
As a result, in practice, the number of elements in $\mathcal{N}_{2\filterdeg}[\signalsupport_t]$ can be considerably lower than its bound 
in \eqref{eq; bound degree}.
This further emphasizes the advantage of reducing computational complexity 
 by searching over $\nei{2\filterdeg}{\signalsupport_t}$ instead of $\nodes$.

Dictionary atoms within the set $\mathcal{N}_{2\Psi}(\signalsupport_t)$ may differ in the significance of the information they provide about the signal. Specifically, some dictionary atoms might convey less relevant information compared to others, with some containing only minor details that are less significant in the context of the signal. 
To remove the less informative dictionary atoms, we define the subset $\mathcal{D}$ as including only those dictionary atoms that offer substantial information about the signal.

\begin{defin} \label{def; D}
Let $\epsilon>0$ be a tuning parameter. Then, the set $\mathcal{D}\subset\mathcal{N}_{2\Psi}(\signalsupport_t)$ is defined such that   
\be \label{eq; D}
\|\Pmat_{\filter_{m}}\filter_{\signalsupport_t}\signal_{\signalsupport_t}\|_2^2<\epsilon,~~~\forall m\in\{\mathcal{N}_{2\Psi}(\signalsupport_t)\setminus\mathcal{D}\}.
\ee 
\end{defin}
The definition of the set $\mathcal{D}$ is determined by the parameter  $\epsilon$, introducing a fundamental trade-off.
A larger value of $\epsilon$ will result in a smaller set $\mathcal{D}$, which may improve the computational efficiency of the proposed sparse recovery algorithms.
However, this can also
 lead to the exclusion of dictionary atoms 
containing significant signal information
 from the set $\mathcal{D}$, thereby degrading the overall recovery performance.
Conversely, selecting a smaller value of $\epsilon$ results in a larger set $\mathcal{D}$, 
potentially reaching the extreme case where $\mathcal{D}=\mathcal{N}_{2\Psi}(\signalsupport_t)$.
In this scenario, the size of the search space 
 and the computational overhead are not effectively reduced. 
In addition, the choice of $\epsilon$ can ensure the existence of a non-empty set $\mathcal{D}$.
Thus, careful tuning of $\epsilon$ is crucial.
One approach is to set $\epsilon$ based on the \ac{snr}, ensuring that atoms with contributions comparable to the noise variance 
are excluded from $\mathcal{D}$.
Another strategy is to leverage prior information, if available, to to guide the selection of  $\epsilon$.

Using the set $\mathcal{D}$ instead of the entire node set $\nodes$ 
for the objective function in \eqref{eq; gic}, i.e. searching over 
\(\{\signalsupport\subset\mathcal{D}:~\card{\signalsupport}\le \sparsityterm\}\),
offers significant advantages in some scenarios, depending on 
the sparsity pattern of the signal.
In these scenarios, 
 reducing the search space may enable the partitioning of 
$\mathcal{D}$ into subsets, \(\{\mathcal{D}_q\}_q\) that are associated with dictionary atoms spanning orthogonal subspaces, as detailed in Section \ref{sec; locality based parition}. 
In this case, the search set can be divided \ac{wrt} the partition
to smaller search sets: \(\left\{\{\signalsupport\subset\mathcal{D}_q:~\card{\signalsupport}\le \sparsityterm\}\right\}_q\).
This strategic partitioning facilitates parallel processing, which enhances the computational efficiency during the search over 
$\mathcal{D}$.

\label{sec; analysis end}
\section{Sparse Recovery of Graph Signals} \label{Sec; theory sparse recovery}
In this section, we propose two methods
for solving the \ac{gic} support set recovery problem in \eqref{eq; gic} 
under the \ac{gsp} measurement model in \eqref{eq; model}. 
The analysis in Section \ref{sec; analysis} provides the basis for these methods, since it enables searching for the optimal support set 
over a reduced set of dictionary atoms. 
In Subsection \ref{sec; gmgic}, we present the \ac{gmgic} method, which involves a node partitioning followed by the implementation of the \ac{gic} locally on each partitioned subset.
In Subsection \ref{sec; sparse; gbgic}, we present the \ac{gbnb} method, which is based on the \ac{bnb} method \cite{boyd2007branch, lawler1966branch},
and involves an upper bound that is derived based on the underlying graphical structure, thus significantly enhancing the traceability of this method.



\subsection{\ac{gmgic}} \label{sec; gmgic}
In this subsection, we outline the rationale and then present the methodology of the \ac{gmgic} method.

\subsubsection{ Node partition} \label{sec; locality based parition}

The analysis in Section \ref{sec; analysis} implies that the true support set of the signal, $\signalsupport_t$, 
is included in the set $\candidateset$ from \eqref{eq; D}.
 We define the following partition on the set $\mathcal{D}$, $\{\mathcal{D}_q\}_{q=1}^Q$, such that 
\be \label{eq; partition}
d_{\mathcal{G}}(k,m)>2\Psi,~\forall k\in\mathcal{D}_q, ~~~\forall m\in\mathcal{D}_{q'},
\ee
for any two subsets $\mathcal{D}_q$ and $\mathcal{D}_{q'}$  from $\{\mathcal{D}_q\}_{q=1}^Q$, where $q\ne q'$ and $1\leq Q \leq s$. The partial support set in the $q$th area, denoted as $\signalsupport_q$, is defined as $\signalsupport_q=\{\signalsupport_t\cap\mathcal{D}_q\}$. 
It should be noted that $\{\signalsupport_q\}_{q=1}^Q$ is a {\em{partition}} of  
the true support set $\signalsupport_t$, since there is no overlap between these subsets and $\signalsupport_t=\cup_{q=1}^Q \signalsupport_q$.

Based on \eqref{eq; F prop 2} in Lemma \ref{lem; F prop}, the partition in \eqref{eq; partition} implies that
any two dictionary atoms associated with two elements from two different subsets are orthogonal, i.e.
\be \label{eq; partition orthogonal}
\filter_{\ntwo}^T\filter_{\none} = 0,~\forall k\in\mathcal{D}_q, ~~~\forall m\in\mathcal{D}_{q'}, 
\ee
for any $q'$ and $q$ where $q\ne q'$. 
Hence, by definition, the column spaces associated with two different subsets are orthogonal, i.e. 
$\text{col}(\filter_{\mathcal{D}_{q}})\perp \text{col}(\filter_{\mathcal{D}_{q'}})$. 
Therefore, the projection of $\filter_{\signalsupport_{q}}\xvec_{\signalsupport_q}$, where $\signalsupport_q=\{\signalsupport_t\cap\mathcal{D}_q\}$, onto $\text{col}(\filter_{\mathcal{D}_{q'}})$ satisfies
\be \label{eq; no affect}
\proj{\mathcal{D}_{q'}}\filter_{\signalsupport_{q}}\xvec_{\signalsupport_q}=\zerovec,
\ee
for any $q'$ and $q$ where $q\ne q'$.
This implies that the partial support set $\signalsupport_q \subset \mathcal{D}_q$
does not affect the measurements in other subsets $\mathcal{D}_{q'}$,  $q'=1,\ldots,Q, q'\neq  q$. 
Thus, instead of searching for the true support set, $\signalsupport_t$, over the entire set $\mathcal{D}$, we aim to find the partial 
support sets, $\{\signalsupport_q\}_{q=1}^Q$, by searching across the smaller subsets $\{\mathcal{D}_q\}_{q=1}^Q$.  
In  Algorithm \ref{alg; partition}
presents an algorithm to partition $\candidateset$
to connected sets that satisfy
 \eqref{eq; partition}.

\begin{center}
\begin{algorithm}[hbt]
\SetAlgoLined
\caption{Node Partitioning}
\label{alg; partition} 
\SetKwInput{Input}{Input}
\SetKwInput{Output}{Output}
\Input{  Graph $\graph$, node subset $\candidateset$}
\BlankLine
Set $\tilde{\nodes}=\candidateset$~~~~ ~~~\% {\em{new node set}}\\
Set ~~~~~~~~~~~~~~~~ \% {\em{new edge set}}
\be
\label{tilde_edge}
\tilde{\edges}= \{(\none,\ntwo):~\none,\ntwo\in\tilde{\nodes},~ 1\le \geo{\none}{\ntwo} \le 2\filterdeg\}
\ee 
Generate a partition of  $\tilde{\mathcal{G}}(\tilde{\nodes}, \tilde{\edges})$ with connected components (e.g. using {\em{conncomp}} in Matlab)\\
\textbf{Return} \textit{subsets $\dsubsets$, where $\dsubset{q}$ includes the nodes in the $q$th connected component obtained from the graph partition.}
\end{algorithm}
\end{center}
\begin{claim} \label{claim; partition}
The subsets $\dsubsets$ obtained from Algorithm \ref{alg; partition} satisfy 
\eqref{eq; partition}.  
\end{claim}

\begin{IEEEproof}
Let $\dsubsets$ be the subsets obtained from Algorithm \ref{alg; partition}, and let $\none$ and $\ntwo$ be selected from different subsets. For the sake of contradiction, assume that $\geo{\none}{\ntwo}\le 2\filterdeg$. 
Since $\none$ and $\ntwo$ satisfy $\geo{\none}{\ntwo}\le 2\filterdeg$, they must be connected by one of the edges in the set $\tilde{\edges}$, i.e. $(\none,\ntwo) \in \tilde{\edges}$, where $\tilde{\edges}$ is defined in \eqref{tilde_edge}. 
However, this contradicts our assumption since $(\none,\ntwo)\in\tilde{\edges}$ only if $\none$ and $\ntwo$ are in the same connected component $\dsubset{q}$. Therefore, we can conclude that when $\none$ and $\ntwo$ are in different subsets from $\dsubsets$, then the geodesic distance between them satisfies $\geo{\none}{\ntwo}>2\filterdeg$. 
Correspondingly, the subsets obtained from Algorithm \ref{alg; partition} satisfy \eqref{eq; partition}. 
\end{IEEEproof}

The node partitioning presented in this subsection is used in the \ac{gmgic} method in Subsection \ref{sec; gmgic method}
in order to reduce the amount of required computations without compromising the estimation performance. 
 When partitioning is not possible (i.e. when $Q=1$),  \ac{gmgic} method still offers an advantage over the \ac{gic} method by performing an exhaustive search over $\mathcal{D}$, which is a smaller subset of the overall node set $\mathcal{V}$.
Moreover, when $Q>1$, partitioning the subset $Q$ into smaller subsets $\{\mathcal{D}_q\}_q$ enables the exhaustive search to be conducted on these smaller subsets, 
as illustrated in Fig. \ref{fig; gmgic}.

\subsubsection{Methodology} \label{sec; gmgic method}

The \ac{gmgic} method consists of four steps.
In the {\bf{first step}}, based on \eqref{14},
we estimate the set $\mathcal{D}$ described in \eqref{eq; D} by applying the \ac{glrt} for each $m\in\mathcal{V}$ on the following binary hypothesis testing problem:
\be
\label{H1_H0}
\begin{cases}
    \mathcal{H}_0:~ \proj{\ntwo}\measurements=\proj{\ntwo}\noise \\    \mathcal{H}_1:~\proj{\ntwo}\measurements=\proj{\ntwo}\filter\signal +\proj{\ntwo}\noise.
\end{cases}
\ee
This hypothesis testing
can be understood as the task of detecting a signal embedded in Gaussian noise.
In this case, the signal is projected onto a smaller subspace, $\proj{\ntwo}\measurements$, under both hypotheses.
Under $\mathcal{H}_1$, the signal is presumed to be present in the subspace, whereas in $\mathcal{H}_0$ 
the subspace encompasses solely the noise.

The \ac{glrt} for the hypothesis testing in \eqref{H1_H0} is defined for each $\ntwo\in\nodes$ as (see, e.g. p. 201 in \cite{kayfundamentals})
\be
\label{GLRT}
\lvert \lvert \proj{\ntwo}   \yvec \rvert\rvert^2 \mathop{\gtrless}_{\Hnull}^{\mathcal{H}_1} \zeta,
\ee
where $\zeta$ is a user-defined threshold parameter.  This threshold is closely related to the parameter $\epsilon$ from the definition of the set $\mathcal{D}$ in \eqref{eq; D}.
However, while $\epsilon$ is associated with the filtered signal $\filter\signal$ (see in \eqref{eq; D}),  
$\zeta$ is applied to the measurement vector $\yvec$ (see in \eqref{GLRT}) and, thus, also accounts for the measurement noise. 
\color{black}
Thus, the set $\candidateset$ is estimated according to the binary detector in \eqref{GLRT},
which is defined for each $\ntwo\in\nodes$, by 
\be  \label{eq; pre-screening}
\hat{\candidateset}\define\{\ntwo\in\mathcal{V}:~\| \proj{\ntwo}  \measurements \|^2  >\zeta\}.
\ee

In the {\bf{second step}}, we use the node partition algorithm in Algorithm \ref{alg; partition} with the input $\hat{\mathcal{D}}$ from \eqref{eq; pre-screening} to obtain the subsets $\hdsubsets$.
In the {\bf{third step}}, we solve the support set recovery problem in \eqref{eq; gic} for each subset $\hat{\mathcal{D}}_q$ separately. 
That is, we estimate 
the partial support set by
\be \label{eq; gic q}
	\hat{\signalsupport}_q = \arg\underset{\{\lvert \signalsupport \rvert \le \sparsityterm,\signalsupport\subseteq \hat{\mathcal{D}}_q\}} {\max} ~\|\proj{\signalsupport}\measurements\|^2-\penaltygic{\signalsupport} ,~~~q=1,\ldots,Q.
\ee
It is noted that for any $q=1,\ldots,Q$, the partial support set, $\hat{\signalsupport}_q$, may be smaller than or equal to the subset from which it is recovered, $\hat{\mathcal{D}}_q$.
The total support set of the sparse vector is then computed as the union of  the partial estimates obtained from \eqref{eq; gic q}, as follows:
\be \label{eq; joint recovered support}
\hat{\signalsupport}^{temp}=\bigcup\nolimits_{q=1}^{Q}\hat{\signalsupport}_q.
\ee
It should be noted that in \eqref{eq; gic q} 
we impose the condition that the size of the support set should not exceed the sparsity parameter, $\lvert \signalsupport \rvert \le \sparsityterm$, for each subset separately.  As a result, we can ensure that the cardinality of $\hat{\signalsupport}_q$ is always less than or equal to $\sparsityterm$ for any $q=1,\ldots,Q$. However, the same restriction does not apply to the union set in \eqref{eq; joint recovered support}, which means that the cardinality of $\hat{\signalsupport}^{temp}$ 
may be greater than $\sparsityterm$, and thus, it may not be in the feasible set of the original optimization problem in \eqref{eq; gic}. In order to address this issue,
in the {\bf{fourth step}}, sparsity-level correction is performed by solving the recovery in \eqref{eq; sup}
for 
$\{ \signalsupport\subset\hat{\signalsupport}^{temp}: \card{\signalsupport}\le \sparsityterm\}$.
That is, we compute the final estimator of the support set by
\be \label{eq; gic sparsity correction}
	\hat{\signalsupport} = \arg\underset{\{\lvert \signalsupport \rvert \le \sparsityterm,\signalsupport\subseteq \hat{\signalsupport}^{temp}\}} {\max} ~\|\proj{\signalsupport}\measurements\|^2-\penaltygic{\signalsupport} . 
\ee
The \ac{gmgic} method is summarized in Algorithm \ref{alg; gmgic} and illustrated in Fig. \ref{fig; gmgic}.
\begin{algorithm}[hbt]
\SetAlgoLined
\caption{GM-GIC}
\label{alg; gmgic}
\SetKwInput{Input}{Input}
\SetKwInput{Output}{Output}
\Input{Measurements $\measurements$, graph filter $\filter$, filter degree $\filterdeg$, 
and sparsity parameter $\sparsityterm$}
\BlankLine
\underline{Step 1:} Compute $\hat{\candidateset}$  from \eqref{eq; pre-screening}\\
\underline{Step 2:} Partition $\hat{\candidateset}$ by Algorithm \ref{alg; partition} to obtain
$\hdsubsets$\ 
\\
 \nonl \underline{Step 3:}
  \\
\For{$q=1$ \KwTo $Q$}{
 Obtain $\hat{\signalsupport}_q$ by solving \eqref{eq; gic q} \
 }
Compute: $\hat{\signalsupport}^{temp}=\bigcup\limits_{q=1}^{Q}\hat{\signalsupport}_q$ \
\\
\nonl\underline{Step 4:}
\\
 Obtain $\hat{\signalsupport}$ by solving \eqref{eq; gic sparsity correction}. 

\textbf{Return} \textit{recovered support: $\hat{\signalsupport}$}
\end{algorithm}

\begin{figure}[hbt]
\centering 
 \includegraphics[trim={0cm 9cm 0 4.8cm},clip,width=7cm]{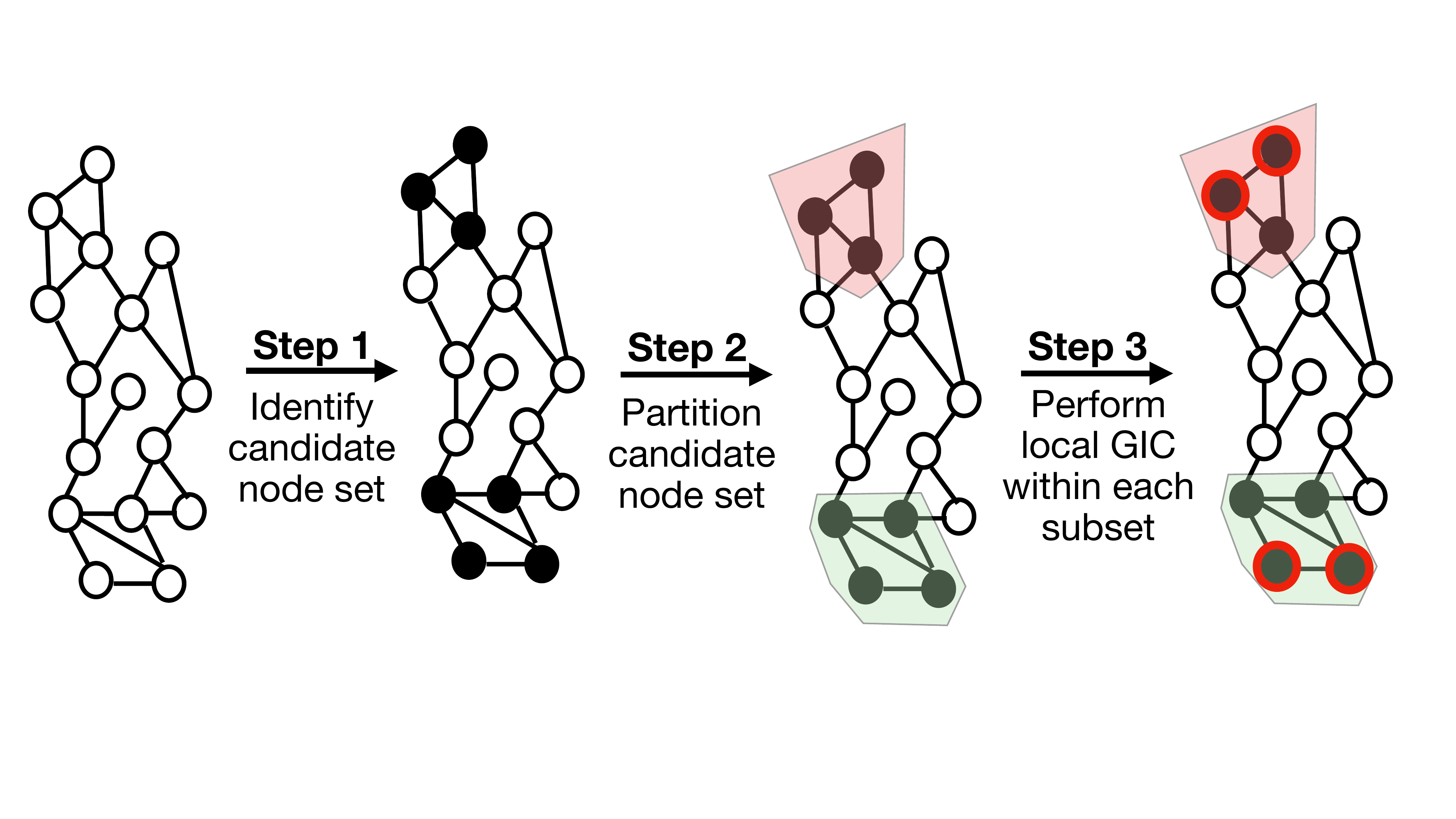}
\caption{ Steps 1-3 of the \ac{gmgic} method from  Algorithm \ref{alg; gmgic}.
In Step $1$, the algorithm estimates the candidate node set,  $\candidateset$,  denoted by black-filled nodes.
A node $\none$ is included in 
 $\candidateset$ if the projection of the measurements into the  single-column space $\text{col}(\filter_{\none})$, $\lvert \lvert \proj{\none} \yvec \rvert\rvert^2$, exceeds a certain threshold. 
In Step $2$, the nodes in $\candidateset$ are partitioned into subsets (here, 2 subsets in pink and green) based on the geodesic distance between the nodes, the graph filter order, and the graph structure. 
Finally, in Step $3$, a local \ac{gic} is performed for the support set recovery within the subsets.  
This results in separate partial support sets (denoted by red rings around the nodes).}
\label{fig; gmgic}
\end{figure}


\subsection{\ac{gbnb}} \label{sec; sparse; gbgic}
In this subsection 
 we present a new \ac{gbnb} method 
for solving the \ac{gic} problem in \eqref{eq; gic}. To this end, 
we reformulate the \ac{gic} problem in \eqref{eq; gic} as an integer programming optimization problem, 
and then derive the \ac{gbnb} method. 
The \ac{gbnb}, is a \ac{bnb}-based method (see, e.g. \cite{boyd2007branch,lawler1966branch}) 
with novel graph-based designs for the branching and bounding steps. 

\subsubsection{New formulation}
By using the $\card{\nodes}\times 1$ Boolean vector $\boovec=[\boo{1},\ldots,\boo{\card{\nodes}}]^T$, 
we address the \ac{gic} problem in \eqref{eq; gic} as the following integer programming optimization problem:
\beqna \label{eq; gic boolean}
\hat{\boovec}=\arg\underset{\boovec\in \{0,1\}^{\card{\nodes}} }{\max}\norm{\proj{\signalsupport(\boovec)}\measurements}^2-\penaltygic{\signalsupport(\boovec)} \nonumber\\
 \text{such that} ~
\card{\signalsupport(\boovec)}\le \sparsityterm,\hspace{2cm}
\eeqna
where $\signalsupport(\boovec)$ is the support set of $\zvec$, i.e. the 
set of indices corresponding to non-zero elements of  $\boovec$ that are in this case equal to 1.
It can be verified that the problem in \eqref{eq; gic boolean} is equivalent to the problem in \eqref{eq; gic} (provides the same solution).
In the new formulation, each node $\none$ can be fixed to either $z_{\none}=0$ or $z_{\none}=1$. 
In this way we can impose additional information on the specific node, 
and determine whether it is part of the support set 
or not. 
The added flexibility is essential for the {\em{branch step}}, which is described next.

\subsubsection{Branch} \label{sec; branch}
Without loss of generality, we assume that the graph node indices are ordered according to their single-node projected energy, i.e. 
\be \label{eq; ordering}
\norm{\proj{1}\measurements}_2^2\ge \norm{\proj{2}\measurements}_2^2\ge \ldots \ge \norm{\proj{\card{\nodes}}\measurements}_2^2.   
\ee
According to the discussion around \eqref{eq; D}, nodes included in the true support set of the signal $\signal$, $\signalsupport_t$, 
are associated with smaller indices in \eqref{eq; ordering}. 
Moreover, dictionary atoms associated with nodes indexed larger than $\card{\candidateset}$, \ac{wrt} to the ordering in \eqref{eq; ordering}, 
are not expected to contain substantial information on the filtered signal, $\filter\signal$. 
Based on the ordering in \eqref{eq; ordering}, we describe the {\em{branch step}} as follows. 
By selecting index $1$, we form the following two problems: 
the first problem is 
\beqna \label{eq; gic boolean 0}
\hat{\boovec}=\arg\underset{\boovec\in \{0,1\}^{\card{\nodes}} }{\max}\norm{\proj{\signalsupport(\boovec)}\measurements}^2-\penaltygic{\signalsupport(\boovec)} \nonumber\\
 \text{such that} \begin{cases}
\card{\signalsupport(\boovec)}\le \sparsityterm \\
\boo{1}=0,
\end{cases}\hspace{2cm}
\eeqna
and the second problem is
\beqna \label{eq; gic boolean 1}
\hat{\boovec}=\arg\underset{\boovec\in \{0,1\}^{\card{\nodes}} }{\max}\norm{\proj{\signalsupport(\boovec)}\measurements}^2-\penaltygic{\signalsupport(\boovec)} \\
 \text{such that} \begin{cases}
\card{\signalsupport(\boovec)}\le \sparsityterm \\
\boo{1}=1.
\end{cases}\hspace{2cm}
\eeqna
In other words, we fix the value of $\boo{1}$ to $0$ in the first subproblem, which means that the first node is not part of $\signalsupport(\boovec)$.
In the second subproblem, we set the value of $\boo{1}$ to be $1$, which means that the first node is part of $\signalsupport(\boovec)$. 
The problems in \eqref{eq; gic boolean 0} and \eqref{eq; gic boolean 1} are subproblems of the original problem in \eqref{eq; gic boolean} since they have the same cost function and constraints as the original problem 
with one additional variable fixed. 
Each of these subproblems is an integer programming optimization
problem, with $\card{\nodes}-1$ Boolean variables. 
The optimal value of the original problem is
the maximum between the optimal values of the two subproblems.

\definecolor{bnbone}{rgb}{0.71,0.89,0.7}
\definecolor{bnbtwo}{rgb}{1,0.64,0.28}

The division of \eqref{eq; gic boolean} to the two subproblems in \eqref{eq; gic boolean 0}-\eqref{eq; gic boolean 1} 
is the first {\em{branch step}} performed in our algorithm. 
The second {\em{branch step}} 
is obtained by splitting one of the above subproblems, once with $\boo{2}=0$ and once with $\boo{2}=1$,
following the order in \eqref{eq; ordering}. 
By continuing in this manner, we form a partial binary tree of subproblems (see our illustration in Fig. \ref{fig; gbnb tree}). 
The root is the original problem depicted in \eqref{eq; gic boolean}. 
This problem is split into the two child subproblems in \eqref{eq; gic boolean 0} and \eqref{eq; gic boolean 1}. 
The second iteration yields another two children of one of the original children.
In general, the subproblems can be written as
\beqna \label{eq; gic boolean case}
\hat{\boovec}=\arg\underset{\boovec\in \{0,1\}^{\card{\nodes}} }{\max}\norm{\proj{\signalsupport(\boovec)}\measurements}^2-\penaltygic{\signalsupport(\boovec)} \nonumber\\
 \text{such that} \begin{cases}
\card{\signalsupport(\boovec)}\le \sparsityterm \\
\boo{\none}=0,~\none\in\mathcal{S}_0 \\
\boo{\none}=1,~ \none\in\mathcal{S}_1,
\end{cases}\hspace{1.25cm}
\eeqna
where $\mathcal{S}_0 $ is the set of Boolean variables fixed to $1$ and $\mathcal{S}_1$ is the set of Boolean variables fixed to $0$.
It can be seen that in the optimization problem in \eqref{eq; gic boolean case}, the equality constraints determine 
$\card{\mathcal{S}_0}+\card{\mathcal{S}_1}$ from the elements of $\zvec$. 
Hence, each subproblem could be represented by the notation $(\mathcal{S}_0,\mathcal{S}_1)$, e.g. $(\emptyset,\emptyset)$,
$(\{1\},\emptyset)$, and $(\emptyset,\{1\})$ represent \eqref{eq; gic boolean}, \eqref{eq; gic boolean 0}, and \eqref{eq; gic boolean 1}, respectively.
Additionally, splitting \eqref{eq; gic boolean case} into two subproblems is achieved by adding node $\card{\mathcal{S}_0}+\card{\mathcal{S}_1}+1$ once to the set $\mathcal{S}_0$
and once to the set $\mathcal{S}_1$. 
Hence, the index used to split a leaf node is determined based on the ordering in \eqref{eq; ordering}, where a leaf node refers to a node in the partial binary tree that does not have children nodes. 
Figure \ref{fig; gbnb tree} illustrates two iterations of the {\em{branch step}}.

\begin{figure}[hbt]
\centering
\definecolor{mycolor}{rgb}{0,0,0} 
\begin{tikzpicture}[level/.style={sibling distance=30mm/#1}, ball/.style={circle, shading=ball, ball color=mycolor,draw=mycolor, fill=mycolor, minimum size=2mm},
line width=1.2pt,scale=0.8]
\node[ball] (root) {} 
  child {
    node[ball] (leftChild) {} 
    edge from parent node[left] {$z_1=0$} 
  }
  child {
    node[ball] (rightChild) {} 
    child {
      node[ball] (rightleftChild) {} 
      edge from parent node[left] {$z_2=0$} 
    }
    child {
      node[ball] (rightrightChild){} 
      edge from parent node[right] {$z_2=1$} 
    }
    edge from parent node[right] {$z_1=1$} 
  };

\node[left, xshift=-1mm, yshift=1mm] at (root) {$\{{\color{bnbone}\emptyset},{\color{bnbtwo}\emptyset\}}$}; 
\node[left, xshift=-1mm, yshift=1mm] at (leftChild) {$\{{\color{bnbone}\{1\}},{\color{bnbtwo}\emptyset}\}$}; 
\node[right, xshift=1mm, yshift=1mm]at (rightChild) {$\{{\color{bnbone}\emptyset},{\color{bnbtwo}\{1\}}\}$}; 
\node[left, xshift=-1mm, yshift=1mm]at (rightleftChild) {$\{{\color{bnbone}\{2\}},{\color{bnbtwo}\{1\}}\}$}; 
\node[right, xshift=1mm, yshift=1mm]at (rightrightChild) {$\{{\color{bnbone}\emptyset},{\color{bnbtwo}\{1,2\}}\}$}; 

\end{tikzpicture}

\caption{Illustration of two iterations of the {\em{branch step}}. 
Each node (black circle) corresponds to an optimization problem in \eqref{eq; gic boolean case}, and is identified by the couple 
$\{{\color{bnbone}\mathcal{S}_0},{\color{bnbtwo}\mathcal{S}_1}\}$, that are assumed to be ${\color{bnbtwo}{\text{included in}}}$/$\color{bnbone}{\text{excluded from}}$ the support set, respectively.  }

\label{fig; gbnb tree}
\end{figure}

\subsubsection{Bound}

When applied to large models, standard \ac{bnb} search algorithms (see e.g. \cite{boyd2007branch}) that use theoretical upper and lower bounds may take substantial amounts of computing time for mixed Boolean convex problems.
Using a heuristic bound
to replace the upper or lower bound might achieve a significant speedup in the \ac{bnb} search procedure.
For example, \cite{dionne1979exact,solanki1998using} showed, for the optimal network design problem, that using a heuristic bound may dramatically reduce the computational complexity with a minor effect on the performance. 
Therefore, based on the analysis in Subsection \ref{sec; analysis}, we propose a heuristic upper bound 
alongside a theoretical lower bound for \eqref{eq; gic boolean case}. 

Based on the analysis in Subsection \ref{sec; analysis}, which stems from Assumptions \ref{A; graph degree}-\ref{A; filter degree},  
we propose the following heuristic upper bound on the optimal objective of \eqref{eq; gic boolean case}:
\beqna \label{eq; bnb upper}
\bound_{ub}(\mathcal{S}_0,\mathcal {S}_1)=\norm{\proj{\mathcal{S}_1}\measurements}^2-\penaltygic{\mathcal{S}_1}\hspace{2.5cm}\nonumber \\
+(s-\card{\mathcal{S}_1})\max\left(\norm{\proj{\none+1}\measurements}^2-\rho(1),0\right), 
\eeqna
where $\none=\card{\mathcal{S}_0}+\card{\mathcal{S}_1}$.
The heuristic bound states that the maximum cost function in \eqref{eq; gic boolean case}
is bounded by the superposition of:
\begin{enumerate}
    \item  The cost function associated with the elements in $\zvec$ that are already fixed as $1$ (the elements in $\mathcal{S}_1$):
$\norm{\proj{\mathcal{S}_1}\measurements}^2-\penaltygic{\mathcal{S}_1}$;
\item The cost function associated with the highest individual contributor within the elements not yet determined, i.e. node $\none+1$ where $\none=\card{\mathcal{S}_0}+\card{\mathcal{S}_1}$, multiplied by the number
of elements that can be added to $\mathcal{S}_1$ without exceeding the sparsity limit:
$(s-\card{\mathcal{S}_1})\max\left(\norm{\proj{\none+1}\measurements}^2-\rho(1),0\right)$.
\end{enumerate}

The lower bound on the optimal solution of \eqref{eq; gic boolean case} is obtained by taking the cost function of one of the elements in the feasible set. 
Specifically, we consider the element $\boovec$ 
in which the undetermined elements
(i.e., the elements in 
$\{\card{\nodes}\setminus\{\mathcal{S}_0\}\cup\mathcal{S}_1\}\}$) 
are set to zero. 
All the elements in $\zvec$ are determined for this case, and the non-zero elements correspond to the set $\mathcal{S}_1$.
Thus, the lower bound for \eqref{eq; gic boolean case} is set to 
\be \label{eq; bnb lower}
\bound_{lb}(\mathcal {S}_1)=\norm{\proj{\mathcal{S}_1}\measurements}^2-\penaltygic{\mathcal{S}_1}.
\ee
Hence, the lower bound for \eqref{eq; gic boolean case} is equivalent to the cost function of the support set $\signalsupport=\mathcal{S}_1$. 
It is noted that the lower bound used is a valid bound that is always lower than the optimal value of the solution of \eqref{eq; gic boolean case}. 
This can be verified by observing the fact that  $\signalsupport(\zvec)=\mathcal{S}_1$ is in the feasible set of the maximization problem in \eqref{eq; gic boolean case}, 
and thus its cost would always be lower than or equal to the optimal cost function. 

\subsubsection{Remarks} We conclude the discussion on the bounds with the following two remarks
 that demonstrate the relationships between the lower and upper bounds in special cases. 
\begin{rem}
By substituting \eqref{eq; bnb lower} in \eqref{eq; bnb upper}
we obtain
\beqna \label{eq; bnb upper 2}
\bound_{ub}(\mathcal{S}_0,\mathcal{S}_1)=\bound_{lb}(\mathcal {S}_1)\hspace{3.75cm}\nonumber\\
+(s-\card{\mathcal{S}_1})\max\left(\norm{\proj{\none+1}\measurements}^2-\rho,0\right).
\eeqna
Thus, if $\card{\mathcal{S}_1}=\sparsityterm$, then the lower bound and upper bound are equal, i.e., $\bound_{ub}(\mathcal{S}_0,\mathcal{S}_1)=\bound_{lb}(\mathcal{S}_1)$.
The rationale behind this result is that in this case, the subproblem already reaches the maximal support set cardinality, $\sparsityterm$, 
and, thus, the associated leaf node in the binary tree should not be further divided. 
\end{rem}
\begin{rem} \label{rem; D}
Due to the ordering in \eqref{eq; ordering}, it can be verified that if $k>\card{\candidateset}$, 
then $k\in \{\mathcal{V}\setminus \candidateset\}$. 
Hence, from \eqref{eq; D} we obtain that $\norm{\proj{k}\Hmat_{\signalsupport_t}\signal_{\signalsupport_t}}< \epsilon$. 
Thus, if the noise variance is negligible compared to the filtered signal, \color{black}
it can be assumed that
\be
\label{assumption_remark}\norm{\proj{\none+1}\measurements}^2-\rho<\epsilon,
\ee
where $\rho$  is the \ac{gic} penalty parameter in \eqref{eq; gic}. 
By substituting the result from \eqref{assumption_remark} in \eqref{eq; bnb upper 2}, we obtain that
the difference between the upper and lower bound in \eqref{eq; bnb upper 2} when
$k>\card{\candidateset}$ satisfies
$$
\bound_{ub}(\mathcal{S}_0,\mathcal{S}_1)-\bound_{lb}(\mathcal{S}_1)< (\sparsityterm-\card{\mathcal{S}_1})\epsilon  <\sparsityterm\epsilon,
$$
where the sparsity parameter $\sparsityterm$ indicates the sparsity level of the signal, $\signal$, and $\epsilon$ is defined in \eqref{eq; D}.  
This result implies that observing support set elements with nodes outside $\candidateset$ cannot improve the support set recovery results. 
Consequently, the result shows that in the worst-case scenario of the algorithm, 
it terminates after observing all support set combinations within the set $\mathcal{D}$ from \eqref{eq; D}.
Hence, in its worst case, the algorithm requires an order of $O\left(\sum_{\iterone=1}^{\sparsityterm}\binom{|\mathcal{D}|}{\iterone}\right)$ flops (see Subsection \ref{sec; complexity} for a comparison with the \ac{gic} and \ac{gmgic} methods). 

\end{rem}

\subsubsection{Algorithm}

The algorithm for implementing the \ac{gbnb} method is initialized by the set $\mathcal{L}^{(0)}$ containing only the root node, which is characterized by the sets $(\mathcal{S}_0, \mathcal{S}_1)=(\emptyset,\emptyset)$. 
At each iteration, branching occurs from the leaf node with the smallest depth, where 
if there are multiple leaves with the smallest depth, then the leaf with the largest upper bound is selected. 
That is, the leaf node is selected by 
\beqna \label{eq; bnb branch}
(\hat{\mathcal{S}}_0,\hat{\mathcal{S}}_1)=\arg\underset{(\mathcal{S}_0,\mathcal{S}_1)\in\mathcal{L}^{(t)}}{\max} 
  ~\bound_{ub}(\mathcal{S}_0,\mathcal{S}_1), \hspace{2.5cm} \nonumber \\
   \text{such that}~ \card{\mathcal{S}_0}+\card{\mathcal{S}_1}=\underset{(\tilde{\mathcal{S}}_0,\tilde{\mathcal{S}}_1)\in\mathcal{L}^{(t)}}{\min } \card{\tilde{\mathcal{S}}_0}+\card{\tilde{\mathcal{S}}_1}.
  \eeqna
The branching is performed by the {\em{branch step}} outlined in Subsection \ref{sec; branch}. 
Hence, to generate $\mathcal{L}^{(t+1)}$, all elements in $\mathcal{L}^{(t)}$ are duplicated excluding the selected leaf node. 
The sets of the selected node, 
$(\mathcal{S}_0,\mathcal{S}_1)$, are then split into two child nodes: $(\{\mathcal{S}_0\cup \none\},\mathcal{S}_1)\}$
and $(\mathcal{S}_0,\{\mathcal{S}_1\cup \none\})$, where $\none=\card{\mathcal{S}_0}+\card{\mathcal{S}_1}$.
The algorithm terminates when $L^{(t+1)}= U^{(t+1)}$, where $U^{(t+1)}$ and $L^{(t+1)}$ denote the maximum upper bound and maximum lower bound across all leaf nodes (nodes in $\mathcal{L}^{(t+1)}$), respectively. 
In this case, the maximum upper bound and maximum lower bound belong to the same leaf,
and the set $\mathcal{S}_1$, which is associated with this leaf, represents the estimated support set (i.e. the algorithm output).  
Additionally, if the upper bound of a leaf is smaller than the global lower bound, $L^{(t+1)}$, then it is pruned (removed from the set $\mathcal{L}^{(t+1)}$) since the cost function associated with any of its potential descendants is assumed lower than its upper bound
and that the global lower bound represents a cost function of a feasible solution. 
The proposed \ac{gbnb} method is summarized in Algorithm \ref{alg; gbnb}. 

\begin{algorithm}[hbt]
\SetAlgoLined
\caption{\ac{gbnb}}
\label{alg; gbnb}
\SetKwInput{Input}{Input}
\SetKwInput{Output}{Output}
\Input{Measurements $\measurements$, graph filter $\filter$, and 
sparsity parameter $\sparsityterm$}

\BlankLine
Index $\nodes$ by \eqref{eq; ordering} \\\vspace{0.1cm}
\textbf{Initialize:} \\\vspace{0.1cm}
$t=0$, $\mathcal{L}^{(0)}=\{(\emptyset,\emptyset)\}$\\\vspace{0.1cm}
Compute: $L^{(0)}=\bound_{lb}(\emptyset)$, $U^{(0)}=\bound_{ub}(\emptyset, \emptyset)$\\\vspace{0.1cm}
\While{$L^{(t)}<U^{(t)}$}{
    Pick $(\mathcal{S}_0,\mathcal{S}_1)$ by \eqref{eq; bnb branch}\\
   Set $\none=\card{\mathcal{S}_0}+\card{\mathcal{S}_1}$, $\mathcal{L}^{(t+1)}=\mathcal{L}^{(t)}$ \\
   \nonl \% Branching $(\mathcal{S}_0,\mathcal{S}_1)$: (lines 8-10)\\
   Set $\mathcal{L}^{(t+1)}=\{\mathcal{L}^{(t+1)}\setminus (\mathcal{S}_0,\mathcal{S}_1)\}$ \\
   Set $\mathcal{L}^{(t+1)}=\{\mathcal{L}^{(t+1)} \cup (\{\mathcal{S}_0\cup \none\},\mathcal{S}_1)\}$\\
   Set $\mathcal{L}^{(t+1)}=\{\mathcal{L}^{(t+1)} \cup (\mathcal{S}_0,\{\mathcal{S}_1\cup \none\})\}$\\
   Compute: $L^{(t+1)}=\arg\underset{(\mathcal{S}_0,\mathcal{S}_1)\in\mathcal{L}^{(t+1)}}{\max} ~\bound_{lb}(\mathcal{S}_1)$\\\vspace{0.1cm}
    Compute: $U^{(t+1)}=\arg\underset{(\mathcal{S}_0,\mathcal{S}_1)\in\mathcal{L}^{(t+1)}}{\max} ~\bound_{ub}(\mathcal{S}_0,\mathcal{S}_1)$\\\vspace{0.1cm}
    Pruning: remove $(\mathcal{S}_0,\mathcal{S}_1)\in\mathcal{L}^{(t+1)}$ if it satisfies $\bound_{ub}(\mathcal{S}_0,\mathcal{S}_1)<L^{(t+1)}$\\
   Update: $t=t+1$
}
\textbf{Return}\textit{ recovered support set $\hat{\signalsupport}=\card{\mathcal{S}_1}$ such that $(\mathcal{S}_0,\mathcal{S}_1)\in\mathcal{L}^{(t)}$ and $\bound_{lb}(\mathcal{S}_1)=L^{(t)}$ }
\end{algorithm}

\subsection{Computational Complexity} \label{sec; complexity}

The main goal of this work is to develop graph-based low-complexity methods 
for support set recovery of sparse graph signals. 
Specifically, we aim to design low-complexity methods that approximate the solution of the
\ac{gic} optimization in \eqref{eq; gic}, since a direct solution requires an exhaustive search with computational complexity that grows 
exponentially with the size of the graph. 
In particular, the exhaustive search requires $\sum_{\iterone=1}^{\sparsityterm}\binom{\card{\nodes}}{\iterone}$ floating-point operations (flops), as can be seen in \eqref{eq; gic}.
The computational complexity of the \ac{gmgic} method from Algorithm \ref{alg; gmgic} is primarily determined by the computations of the local \ac{gic} on the largest subset, and thus requires  $O\left(\sum_{\iterone=1}^{\sparsityterm} \binom{ \max_q|\mathcal{D}_q|}{\iterone}\right)$ flops.  In particular, when the set $\mathcal{D}$ cannot be partitioned, the complexity becomes $O\left(\sum_{\iterone=1}^{\sparsityterm} \binom{\mathcal{D}}{\iterone}\right)$. \color{black}
Hence, its complexity is significantly lower than those of the \ac{gic} methods, 
where $\forall q: \mathcal{D}_q\subseteq \mathcal{D}$ and from \eqref{eq; D} and Theorem \ref{thm; graph degree}, $\card{\mathcal{D}}\ll\card{\nodes}$.
The \ac{gbnb} method from Algorithm \ref{alg; gbnb} is anticipated to have an even lower computational burden compared with the \ac{gmgic} method, due to its iterative nature. 
However, as shown in Remark \ref{rem; D}, in its worst-case scenario
it requires an order of $O\left(\sum_{\iterone=1}^{\sparsityterm}\binom{|\mathcal{D}|}{\iterone}\right)$ flops.  
Consequently, even in the worst case, the computational complexity of the \ac{gbnb} method is significantly lower than that of the \ac{gic} method. 
This analysis highlights that the main advantage of the proposed methods lies in their reliance on the size of the set 
$\mathcal{D}$ for complexity considerations, rather than the size of the entire node set, $\mathcal{D}$, thereby enabling efficient computation even in large-scale applications.


\section{\ac{gfoc}} \label{sec; sparse; gfoc}

In this section, we present the \ac{gfoc} method, which
 is a new method for correcting existing estimated support sets of the sparse graph signal, $\hat{\signalsupport}$.
 This correction is obtained by optionally replacing 
 the support set elements with nodes from their one-hop neighborhood,  
 $\firstnei{\hat{\signalsupport}}$. 
The initial estimated support set, $\hat{\signalsupport}$, can be obtained by any sparse recovery method, including standard techniques
(such as the \ac{omp} \cite{cai2011orthogonal} and Lasso \cite{tibshirani1996regression} methods), as well as the \ac{gmgic} and \ac{gbnb} methods proposed in  Section \ref{Sec; theory sparse recovery}.\color{black} 

\subsection{Method}
The proposed \ac{gfoc} method receives as an input a support set estimate, $\hat{\signalsupport}$, provided by any sparse recovery algorithm, along with the system measurements and the information regarding the underlying graph. 
Then, the \ac{gfoc} method goes over the estimated support set elements.
For each element (i.e. each node $k\in\hat{\signalsupport}$), it examines whether the \ac{gic} increases when the element is replaced with one of the nodes in the one-hop open neighborhood set 
\be \label {eq; open neighborhood}
\firstopennei{\none}\define \{\nei{1}{\none}\setminus\none\}=\{\ntwo\in\nodes: \geo{\none}{\ntwo}=1 \},
\ee
where $\nei{1}{\none}$ is defined in \eqref{nei_def} (with $\Delta=1$).
If the \ac{gic} increases, then the method propagates by replacing the $k$th support set element with its neighbor, i.e. the element in \eqref{eq; open neighborhood} that maximizes the \ac{gic}. Otherwise, the $k$th element stays in the estimated support.
It is noted that if the \ac{gfoc} method input is infeasible, i.e. $\card{\hat{\signalsupport}}>s$,
then the input can be corrected by selecting the indices associated with the $s$ largest values in $\hat{\signal}$ where $\hat{\signal}=\psudo{\submatcol{\filter}{\hat{\signalsupport}}}$.
The \ac{gfoc} method is summarized in Algorithm \ref{alg; gfoc}.
In addition, the following example
illustrates the process of the \ac{gfoc} method for a $9$-node graph.

\begin{exm} \label{exm; gfoc}
Fig. \ref{fig; gfoc} illustrates the following scenario: the \ac{gfoc} method receives as input a set of measurements $\measurements$, a graph filter $\filter$, and an estimated support set $\hat{\signalsupport}=\{3,7\}$. 
The \ac{gfoc} method starts by examining the first-order neighborhood of node $3$ (the orange area). 
Specifically, it finds $k$ such that $\{k,7\}$ is the optimal support set that maximizes the \ac{gic}
cost function in \eqref{eq; gic} within the sets: $\{3,7\},$ $\{1,7\},$ $\{4,7\},$ and $\{6,7\}$. 
Similarly, it examines the first-order neighborhood of node $7$ (the green area). 
Hence, it finds $m$ such that $\{k,m\}$ is the optimal support set that maximizes the \ac{gic}
cost function in \eqref{eq; gic} within the sets $\{\none,7\}$, $\{\none,4\},$ $\{\none,5\},$ and $\{\none,9\}$. 
\end{exm}

\begin{figure}[h!]
\centering 
 \includegraphics[trim={20.5cm 3.6cm 22.8cm 5.2cm},clip,width=2.5cm]{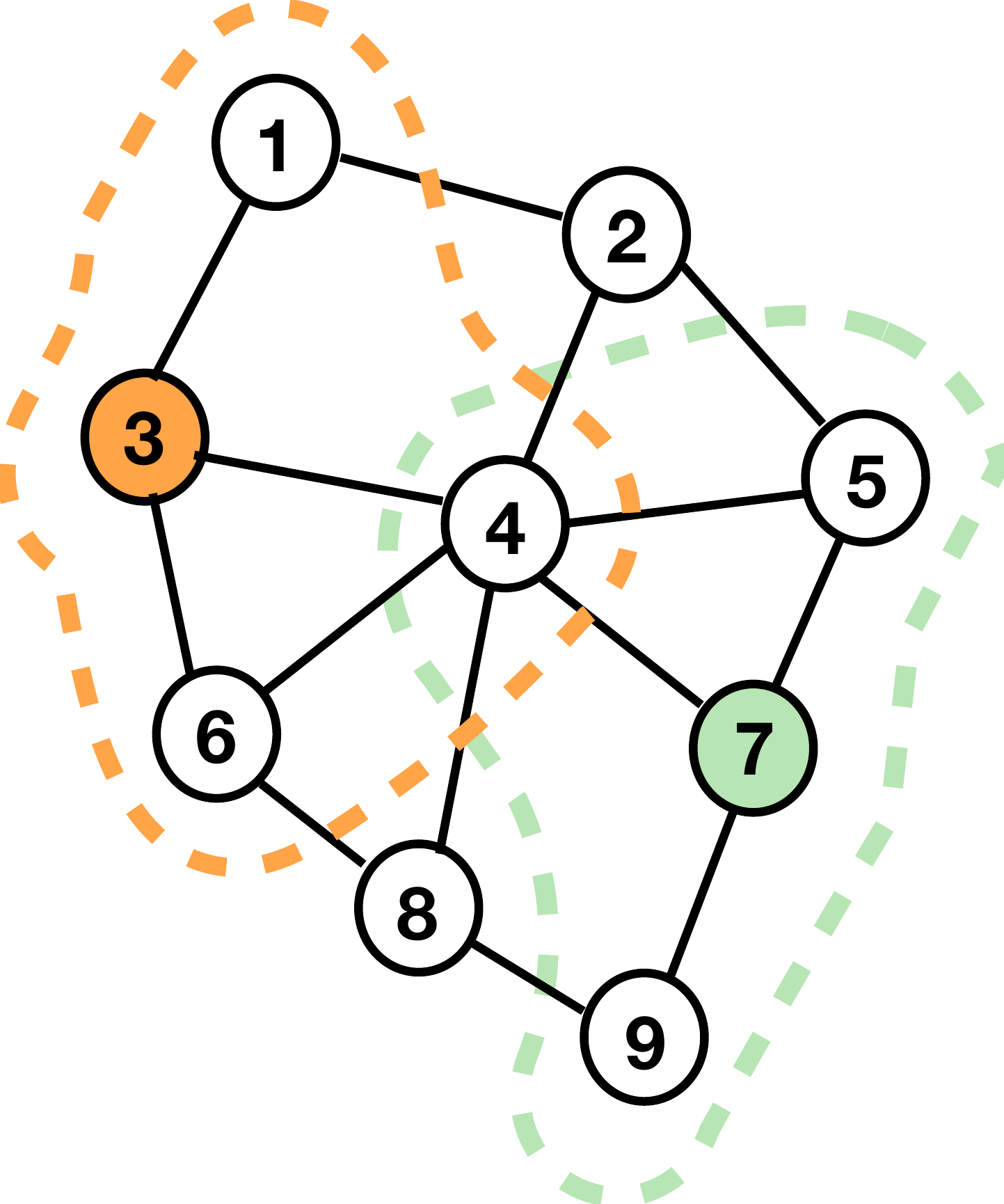}
\caption{ $9$-node graph illustration: the one-hop neighborhoods of nodes $3$ and $7$,
include the nodes inside the orange and green areas, respectively.}
\label{fig; gfoc} 
\end{figure}

\newcommand{\obj}{{obj}}
\begin{algorithm}[hbt]
\SetAlgoLined
\caption{\ac{gfoc}}
\label{alg; gfoc}
\SetKwInput{Input}{Input}
\SetKwInput{Output}{Output}
\Input{Measurements $\measurements$, estimated support set $\hat{\signalsupport}$, graph $\graph$, and graph filter $\filter$ }

Define $\obj(\signalsupport)=\norm{\proj{\signalsupport}\measurements}^2-\penaltygic{\signalsupport} $ \\
Set $\hat{\signalsupport}^c=\hat{\signalsupport}$\\ 
\For{$\none\in \hat{\signalsupport}$}{
    Compute $\hat{\ntwo}=\arg\optmax{\ntwo\in\firstopennei{\none}} \obj(\{\hat{\signalsupport}^c\setminus\none\}\cup\ntwo)$ \\
    \If{$\obj(\{\hat{\signalsupport}^c\setminus\none\}\cup\hat{\ntwo})>\obj(\hat{\signalsupport}^c)$ }
    {
        $\signalsupport^c=\{\{\signalsupport^c\setminus \none\}\cup \hat{\ntwo}\} $
    }    
}
\textbf{Return} \textit{corrected support: $\signalsupport^c$}
\end{algorithm}

\subsection{Discussion} \label{sec; gfoc discussion}

The \ac{gfoc} method is derived based on the \ac{gic} in \eqref{eq; gic} and the analysis in Section \ref{sec; analysis}, derived following Assumptions \ref{A; graph degree}-\ref{A; filter degree}. 
Specifically, based on the discussion around \eqref{eq; D}, it can be inferred that when estimating the support set, if some errors occur, they are likely to be caused by the accidental inclusion of nodes in close proximity to the elements in the true signal support set $\signalsupport_t$. These mistakenly included nodes may be in addition to, or instead of, some of the elements in $\signalsupport_t$.
In response to this possible error, Algorithm \ref{alg; gfoc} proposes examining the immediate neighborhood of elements within the estimated support set to potentially enhance solutions \ac{wrt} the \ac{gic}. 
It is noted that using the \ac{gfoc} method may only improve the performance compared to the method's input \ac{wrt} to the \ac{gic}. 
Moreover, it can be verified from Algorithm \ref{alg; gfoc} that the number of $\gic$ computations required is limited by $\sparsityterm d_{max}$, 
where $d_{max}$ represents the maximum degree of the graph. 
Now, due to Assumption \ref{A; graph degree}, every node in the graph is connected to a small number of neighbors. Hence, we can assume that $d_{\max}$
is relatively low, i.e. $d_{\max}\ll\card{\nodes}$.
Thus, the computational complexity of this approach is expected to be low.
It should be noted that the support of the 
output graph signal obtained by the \ac{gfoc} method can be smaller than the input support set,
as a single neighbor may replace two elements.
The \ac{gfoc} method can be easily extended to the graph $\Psi$ order correction  
by replacing 
\eqref{eq; open neighborhood} with
\be \label {eq; open neighborhood psi}
\{\nei{\Psi}{\none}\setminus\none\}=\{\ntwo\in\nodes: 0<\geo{\none}{\ntwo}\le \Psi \}. 
\ee


\section{Simulations} \label{sec; simulations}

In this section, we evaluate the performance of the proposed support set recovery methods 
developed in Section \ref{Sec; theory sparse recovery} and Section \ref{sec; sparse; gfoc}, and compare them with existing methods through numerical simulations.
The compared methods and the evaluation measures are presented in Subsection \ref{sec; sim; methods and measures}.
We conduct three types of experiments to demonstrate the performance 
of our sparse recovery methods:
1) on
synthetic data over random \ac{sbm} graphs with known graph filters in Subsection \ref{sec; sim; syntetic start}; 2) for blind identification of graph signals over brain networks in Subsection \ref{sec; brain start}; and 3) for recovering the source of the $1854$ Cholera outbreak in London in Subsection \ref{sec; cholera start}.

\subsection{Methods and Performance Measures} \label{sec; sim; methods and measures}
The following support set recovery methods are compared:
\begin{itemize}
\item \label{M; omp} The \ac{omp} method \cite{cai2011orthogonal};
\item \label{M; lasso} The Lasso method \cite{tibshirani1996regression},  which is implemented for diffused sparse graph signals in \cite{ramirez2021graph}; 
\item \label{M; L1BNB} The \ac{bnb} applied on \eqref{eq; SR 2} with the lower and upper bounds suggested in \cite{boyd2007branch} for 
Boolean convex problems (denoted as $\ell_1$-\ac{bnb}); 
\item \label{M; gnmgic} The \ac{gmgic} method in Algorithm \ref{alg; gmgic};
\item \label{M; gbnb} The \ac{gbnb} method in Algorithm \ref{alg; gbnb} (denoted as \gbnb);
\item \label{M; gic} An exhaustive search over the feasible set in \eqref{eq; gic} that involves comparing the cost functions
of all the optional support sets (denoted as \ac{gic}); 
\item The \ac{gfoc} support set correction method in Algorithm \ref{alg; gfoc}, with the estimated support set input from each of the above mentioned methods;
The corrected outputs are denoted as 
1) $\overline{\ac{omp}}^{c}$,
2) $\overline{\text{Lasso}}^{c}$,
3) $\overline{\ell_1-BnB}^{c}$,
4) $\overline{\ac{gmgic}}^{c}$,
5) $\overline{\text{\gbnb}}^{c}$,
6) $\overline{\ac{gic}}^{c}$.
\end{itemize}
In practice, the Lasso and $\ell_1$-\ac{bnb} methods directly solve the sparse recovery problem in \eqref{eq; SR 2}.
As shown in Section \ref{sec; problem}, this is equivalent to solving \eqref{eq; sup}. 
The tuning factors for the different methods are stated as follows. 
The Lasso and $\ell_1$-\ac{bnb} regularization parameter is set to $0.01$. 
Additionally, the \ac{gic} penalty function, $\penaltygic{\signalsupport}$, used for the \ac{gmgic}, \ac{gic}, \gbnb, \ac{gic}, and \ac{gfoc} methods, is chosen to be the \ac{aic} penalty function \cite{stoica2004model},  which demonstrated efficacy across a variety of settings we evaluated.
Hence,  we set 
$\penaltygic{\signalsupport}=2\card{\signalsupport}$.
For the \ac{gmgic} method, the pre-screening threshold $\zeta$ in \eqref{GLRT} is set to $\sigma_n$.
For each presented scenario, the performance has been evaluated by at least $10^3$ Monte Carlo simulations.

The following performance measures are used in the following to evaluate the different methods:
\begin{itemize}
\item The support set recovery is
measured by the F-score classification metric \cite{sokolova}.
The F-score compares between the 
the true support set, $\signalsupport_t$, and the estimated support, $\hat{\signalsupport}$, which is
obtained by the different methods above.
The F-score is given by
\begin{equation} \label{eq; F score}
FS(\signalsupport_t,\hat{\signalsupport})=\frac{2t_p}{2t_p+f_n+f_p},
\end{equation}
where $t_p\define\card{\signalsupport_t\cap\hat{\signalsupport}}$ is the number of true-positives, 
 $f_p\define\card{\{\nodes\setminus\signalsupport_t\}\cap\hat{\signalsupport}}$ is the number of false-positives, 
 and $f_n\define\card{\signalsupport_t\cap\{\nodes\setminus\hat{\signalsupport}\}}$  is the number of false-negatives.
The F-score is between $0$ and $1$, where $1$ indicates perfect support set recovery, i.e. $\hat{\signalsupport}=\signalsupport$.
\item  The sparse signal recovery accuracy is measured using the \ac{mse} between the true graph signal, $\signal$, and the estimated sparse graph signal, $\hat{\signal}$:
$
MSE(\signal, \hat{\signal}) = \left\| \signal - \hat{\signal} \right\|_2^2,
$
where both \(\signal\) and \(\hat{\signal}\) are normalized. 
In practice, the sparse signal is recovered by substituting \(\hat{\signalsupport}\) in \eqref{eq; WLS}. 

\item The computational complexity is evaluated by the average run-time of the algorithms, where the simulations have been conducted  
using Matlab on two Intel(R) Xeon(R) CPU E5-2660 v4 @ 2.00 GHz processors (Subsection \ref{sec; sim; syntetic start})  or using Matlab on the Apple M1 pro chip (Subsections \ref{sec; brain start} and \ref{sec; cholera start}).
\end{itemize}

\begin{figure*}[t]
    \centering
    \begin{subfigure}[b]{0.3\textwidth}
        \centering
        \includegraphics[width=0.7\textwidth]{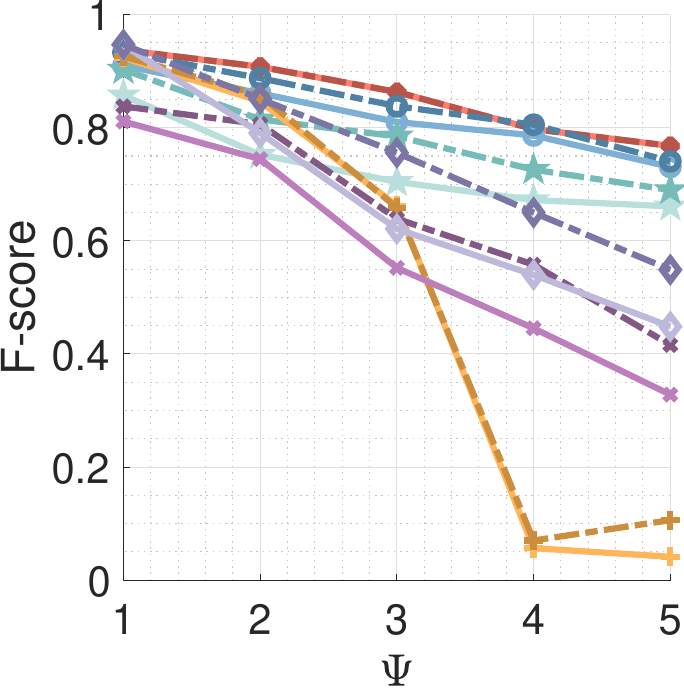}
        \begin{picture}(0,0)
            \put(-125,110){\textbf{(a)}}
        \end{picture}
        \label{fig: fscore deg}
    \end{subfigure}%
    \begin{subfigure}[b]{0.3\textwidth}
        \centering
        \includegraphics[width=0.7\textwidth]{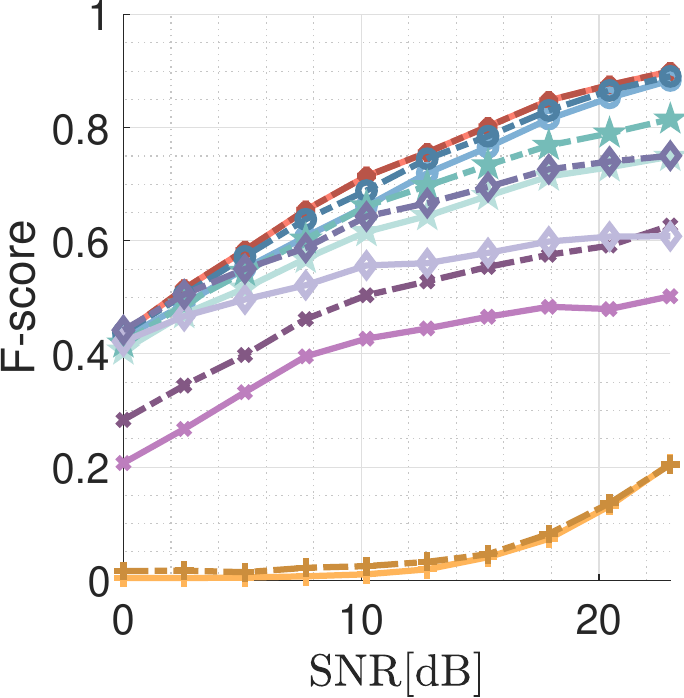}
        \begin{picture}(0,0)
            \put(-125,110){\textbf{(b)}}
        \end{picture}
        \label{fig: fscore snr}
    \end{subfigure}%
    \begin{subfigure}[b]{0.3\textwidth}
        \centering
        \includegraphics[width=0.7\textwidth]{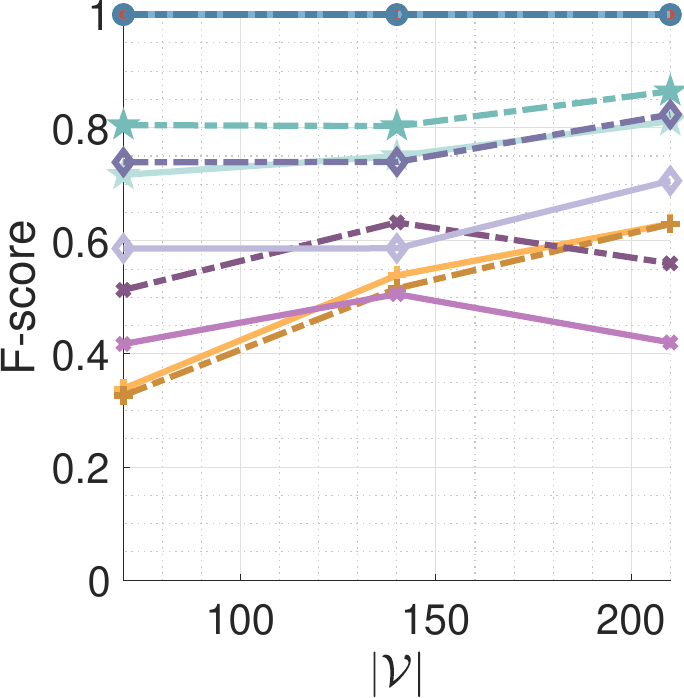}
        \begin{picture}(0,0)
           \put(-125,110){\textbf{(c)}}
        \end{picture}
        \label{fig: fscore graph size}
    \end{subfigure}
\vspace{0.2cm}

    \begin{subfigure}[b]{0.3\textwidth}
        \centering
        \includegraphics[width=0.7\textwidth]{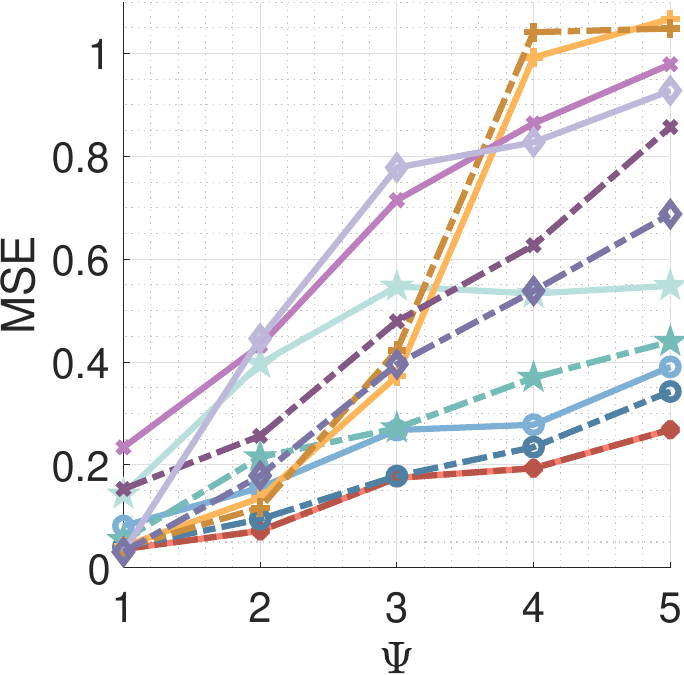}
        \begin{picture}(0,0)
           \put(-125,110){\textbf{(d)}}
        \end{picture}
        \label{fig: mse deg}
    \end{subfigure}%
    \begin{subfigure}[b]{0.3\textwidth}
        \centering
        \includegraphics[width=0.7\textwidth]{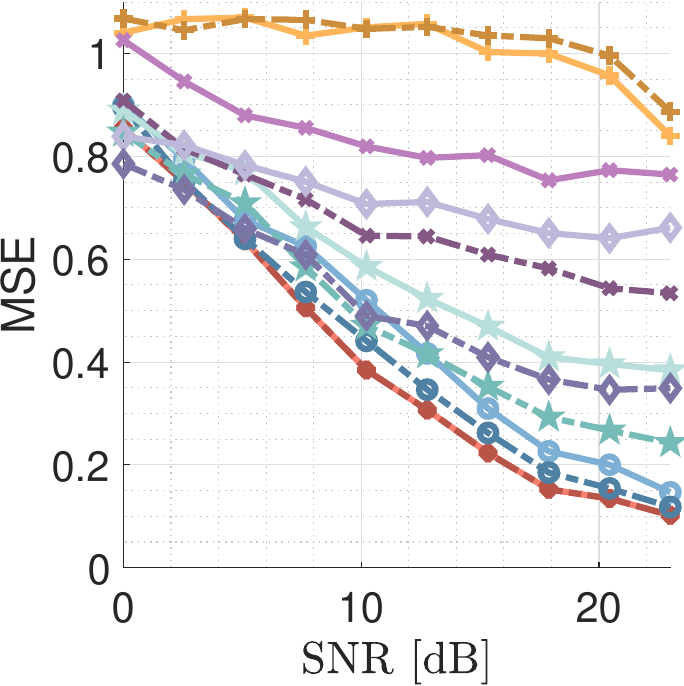}
        \begin{picture}(0,0)
           \put(-125,110){\textbf{(e)}}
        \end{picture}
        \label{fig: mse snr}
    \end{subfigure}%
    \begin{subfigure}[b]{0.3\textwidth}
        \centering
        \includegraphics[width=0.7\textwidth]{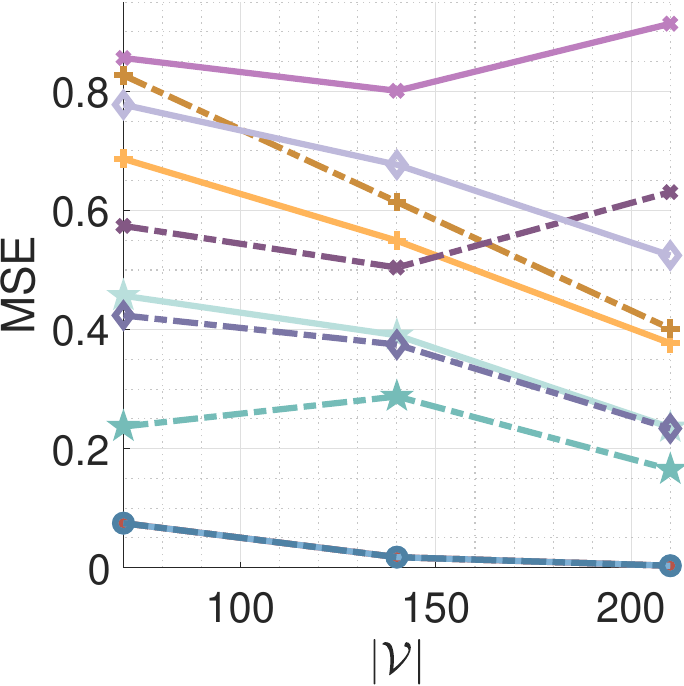}
        \begin{picture}(0,0)
            \put(-125,110){\textbf{(f)}}
        \end{picture}
        \label{fig: mse graph size}
    \end{subfigure}
\vspace{0.2cm}

    \begin{subfigure}[b]{0.3\textwidth}
        \centering
        \includegraphics[width=0.7\textwidth]{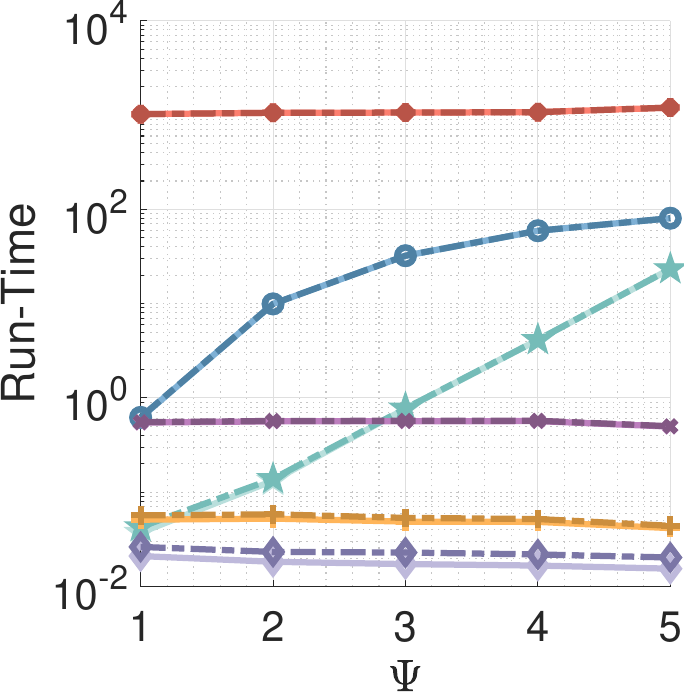}
        \begin{picture}(0,0)
           \put(-125,110){\textbf{(g)}}
        \end{picture}
        \label{fig: run deg}
    \end{subfigure}%
    \begin{subfigure}[b]{0.3\textwidth}
        \centering
        \includegraphics[width=0.7\textwidth]{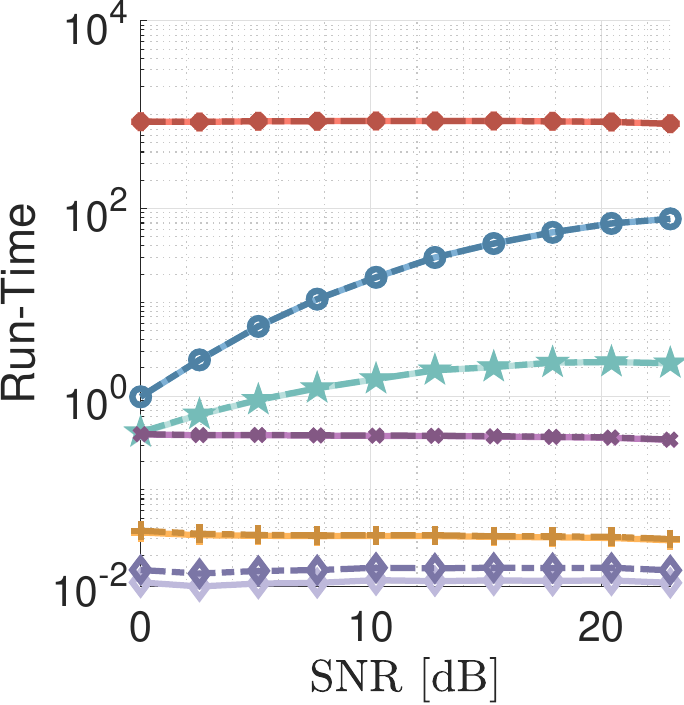}
        \begin{picture}(0,0)
            \put(-125,110){\textbf{(h)}}
        \end{picture}
        \label{fig: run snr}
    \end{subfigure}%
    \begin{subfigure}[b]{0.3\textwidth}
        \centering
        \includegraphics[width=0.7\textwidth]{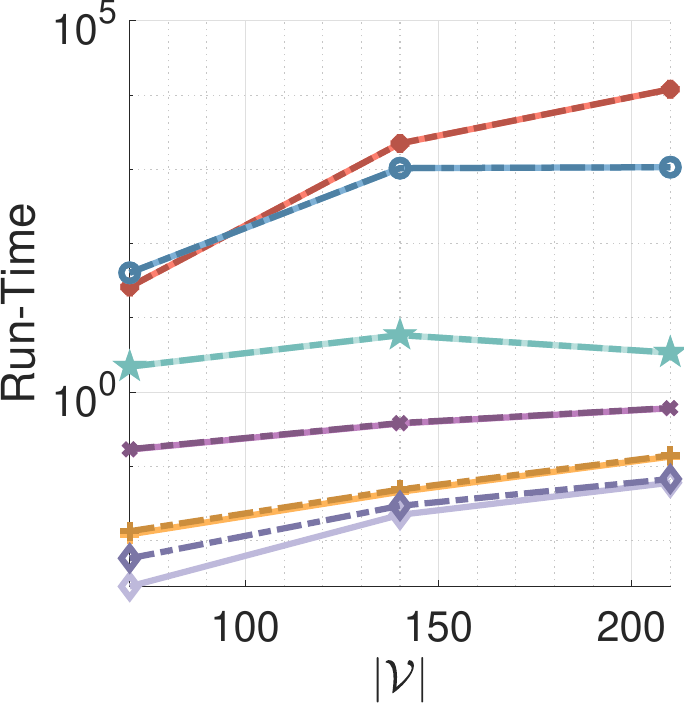}
        \begin{picture}(0,0)
           \put(-125,110){\textbf{(i)}}
        \end{picture}
        \label{fig: run graph size}
    \end{subfigure}

\vspace{0.25cm}

    \begin{minipage}{0.9\textwidth}
        \centering
        \includegraphics[width=0.8\textwidth]{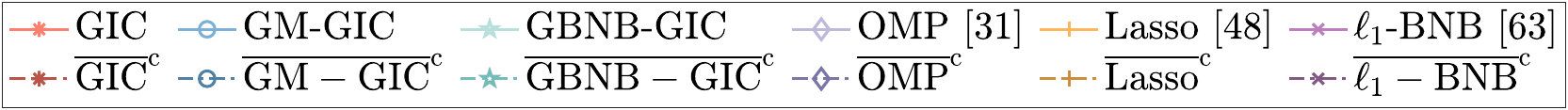} 
    \end{minipage} 
    \caption{ SBM graph model: The F-score ((a)-(c)),  MSE ((d))-(f)), and run-time ((g)-(i)) of the different methods are compared versus the graph filter degree ($\Psi$), SNR, and  graph size. } 
    \label{fig; sbm}
\end{figure*}

\subsection{Test Case A - Synthetic Data Set} \label{sec; sim; syntetic start}

In this subsection, we consider the support set recovery problem in \eqref{eq; sup}, where the sparse dictionary matrix is selected as the 
graph filter defined by \eqref{eq; graph filter}.
The \ac{gso} used to construct the graph filter is the graph Laplacian matrix \cite{shuman2013emerging}, i.e. $\Smat=\Lmat$, which is normalized such that its largest eigenvalue is $12$.
The graph filter coefficients are set to $1$, $h_{\iterone}=1$, $\iterone=0,\ldots\filterdeg$.
Additionally, the measurement vector $\measurements$ 
follows the model in \eqref{eq; model}. 
The underlying graph is modeled as an \ac{sbm} graph with \( \clusternum \) clusters, where each cluster contains \( N=70 \) nodes  \cite{lee2019review}. The probability of connection between any two nodes within the same cluster is \( \frac{6}{N} \). 
In addition, consecutive clusters are connected by $2$ link nodes from each cluster. 
Unless stated otherwise, in the following simulations the \ac{sbm} is composed of $2$ clusters and the support set elements in $\signalsupport$, where $\card{\signalsupport}=4$, are selected from one cluster
 using one of the following two cases: 
\begin{enumerate}[label=-, leftmargin=0.3cm,topsep=2.5pt] 
    \item \label{sensor; 1} 
\textbf{Scenario 1 (Mixed Dispersed-Localized Support):} In this scenario, 
  two nodes are randomly selected from the  cluster, 
    each has one of its neighbors added to the support set. This process is repeated until four different nodes are selected, which results in two groups of four distinct nodes. 
   This scenario describes a situation where the non-zero elements are clustered in two separated groups, 
   where the nodes within each group are in close proximity. 
     \item \label{sensor; 2} \textbf{Scenario 2 (Localized Support):} In this scenario, a single node is randomly selected from the cluster. 
     Then,  $\card{\signalsupport}-1$ ($3$ when $\card{\signalsupport}=4$)  of the node's neighbors are randomly chosen and added to the support set.
     This scenario describes a situation where the non-zero elements are in close proximity. 
  \end{enumerate}
The non-zero elements of the sparse graph signal, $\signal$, are drawn from a standard normal distribution. 
The filtered signal, $\filter\xvec$, is normalized to a specified \ac{snr},  which is calculated as
$\text{SNR}=\norm{\filter\signal}^2/(\card{\nodes}\sigma_n^2)$, where  the noise covariance matrix is $\Rmat=\sigma_n^2\Imat$, and 
unless otherwise specified, $\sigma_n=0.01$.

Figures \ref{fig; sbm}.(a)-(c) present the F-score as a function of: (a) the graph filter degree for SNR$=17$ [dB] and $\card{\nodes}=140$; (b) \ac{snr} for $\filterdeg=4$ and $\card{\nodes}=140$; and (c) and graph size for $\filterdeg=4$ and SNR=$20$ [dB]. 
Across all figures, 
the \ac{gmgic} and its correction $\overline{\ac{gmgic}}^{c}$ methods obtain the closest result  to the \ac{gic} benchmark,  showing the highest F-scores.  
Close behind is the $\overline{\text{\gbnb}}^{c}$ method, with
results up to $0.1$ better than the \gbnb~method in Fig. \ref{fig; sbm}.(a).  The results also show that \ac{gfoc} significantly enhances the $\ell_1$-BnB method across all experiments.

In Fig. \ref{fig; sbm}.(a), it can be seen that the F-score of all methods decreases almost always as the graph filter degree increases.
The $\overline{\ac{omp}}^{c}$ method shows similar results to the $\overline{\text{\gbnb}}^{c}$ method for the first $2$ graph degrees, 
yet lower F-scores at the higher degrees. 
The Lasso and $\ell_1$-\ac{bnb} methods and their corrections show competitive results only for the first $2$ graph degrees. 
In Fig. \ref{fig; sbm}.(b)
it can be seen that the F-score of all methods increases as 
the SNR increases. 
The corrected version of the \ac{omp} method, $\overline{\ac{omp}}^{c}$, 
improves the \ac{omp} F-score by approximately $25\%$ compared to the original \ac{omp}. 
 Similarly, the corrected version of the $\ell_1$-\ac{bnb}, $\overline{\ell_1-BnB}^{c}$ method,  improves the F-score of the $\ell_1$-\ac{bnb} by up to $20\%$.
 However, 
it remains less competitive than the $\overline{\ac{omp}}^{c}$, $\overline{\text{\gbnb}^{c}}$, and \ac{gmgic}
methods. 
The Lasso method and its corrected version
exhibit poor F-score performance across all observed \acp{snr}. 

Figure \ref{fig; sbm}.(c) displays the F-score as a function of graph size (number of \ac{sbm} clusters). For $\card{\nodes}=70$, the graph contains one cluster. When $\card{\nodes}=140$, it has two clusters, with the signal support drawn from Cluster $1$. For $\card{\nodes}=210$, the graph has three clusters, with the support either entirely from Cluster $1$ or uniformly from Clusters $1$ and $3$. In all cases, the support follows Scenario $2$, and the \ac{snr} is set to SNR=$20$ dB.
It can be seen that the graph size has a minimal impact on the F-score of the proposed methods. 
The main effect is seen with the $\ell_1$ and Lasso methods (and their corrections), which perform worse compared to the other approaches. 

Figures \ref{fig; sbm}.(d)-(f) presents the \acp{mse} of the different methods as a function of the (a) graph filter degree, (b) the \ac{snr}, and (c) the graph size, using the same settings as those used for the F-score evaluation in Figs. \ref{fig; sbm}.(a)-(c).
The \ac{mse} performance in Figs. \ref{fig; sbm}.(d)-(f) aligns closely with the F-scores results in Figs. \ref{fig; sbm}.(a)-(c), leading to similar conclusions regarding the comparative performance of the different methods and the superiority of the proposed methods.
For example, the \ac{gic}, \ac{gmgic}, and $\overline{\text{\gbnb}^{c}}$ methods, which achieve the highest F-scores in Figs. \ref{fig; sbm}.(a)-(c), also exhibit the lowest \acp{mse} in Figs. \ref{fig; sbm}.(d)-(f).
In addition, it can be seen in Fig. \ref{fig; sbm}.(e) that the \ac{snr} significantly impacts the performance of the \ac{gic}, \ac{gmgic}, and $\overline{\text{\gbnb}^{c}}$ methods, showing improvements of approximately $0.7$ across the tested range. In contrast, the \ac{omp} method improves by only $0.1$ and shows no further improvement beyond an \ac{snr} of $8$ [dB]. Similarly, $\overline{\ac{omp}}^{c}$ improves by around $0.4$ and shows no further improvement beyond an \ac{snr} of $14$ [dB].

In Figs. \ref{fig; sbm}.(g)-(i) we present the run-time of the different methods as a function of the (g) graph filter degree, (h) the \ac{snr}, (i) and the graph size, respectively,  for the same setting as used for the F-score and \ac{mse} evaluation in Figs. \ref{fig; sbm}.(a-c) and  Figs. \ref{fig; sbm}.(d)-(f).
Across all figures, the \ac{gic} benchmark method requires the highest run-time,  exceeding $1,000$ seconds per Monte-Carlo simulation when the graph size is greater than $140$ nodes. 
In contrast, the proposed \ac{gmgic} and \gbnb~methods have significantly reduced run-times while still maintaining good F-scores and \ac{mse} results, as shown in Figs. \ref{fig; sbm}.(a)-(c) and Figs. \ref{fig; sbm}.(d)-(f).
The  $\ell_1$-\ac{bnb},  Lasso, and \ac{omp} methods are the fastest overall. However, as shown in Figs. \ref{fig; sbm}.(a)-(c)  and Figs. \ref{fig; sbm}.(d)-(f),  these methods display lower F-scores  and \ac{mse}, respectively,  than the proposed methods. 
Notably, the additional run-time required by the \ac{gfoc} method to correct these methods is negligible.

In Fig. \ref{fig; sbm}.(g) it can be seen that the run-time of the \ac{gmgic} and \gbnb~methods increases as the 
graph filter degree increases, where the \ac{gmgic} method run-time is larger than the \gbnb~method run-time. 
In Fig.~\ref{fig; sbm}.(h), it can be seen that the run-time of the \ac{gmgic} method significantly increases as the \ac{snr} increases. 
The run-time of the \gbnb~method moderately increases and does not change for \ac{snr} higher than $15$ [dB]. 
Both methods still operate significantly faster than the \ac{gic} method. 
The run-time required for the remaining methods is constant \ac{wrt} the \ac{snr}
 and lower than the methods mentioned above. 
In Fig. \ref{fig; sbm}.(i) it can be seen that
the run-time required for the \ac{gic} method
grows exponentially with the graph (model) size.
 In comparison, the \ac{gmgic} run-time grows significantly between $\card{\mathcal{\nodes}}=70$ to $\card{\mathcal{\nodes}}=140$, but remains constant between $\card{\mathcal{\nodes}}=140$ to $\card{\mathcal{\nodes}}=210$.
The \gbnb~ method is only marginally affected by graph size. 
In addition, the graph size also affects the run-time of the Lasso and \ac{omp} methods, where their run-time becomes closer to those of the \ac{bnb} methods as $|\mathcal{V}|$ increases. 
To conclude, the results demonstrate that the proposed \ac{gmgic} and \ac{gbnb} scale well with the graph size and enable efficient processing even in larger network configurations.

\label{sec; sim; syntetic end}

\subsection{Test Case B - Graph Blind Deconvolution of Brain Signals} \label{sec; brain start}
Graph blind deconvolution has various real-world applications \cite{segarra2016blind,ramirez2021graph}.
In particular, it has been used for the analysis of epileptic 
seizures \cite{kramer2008emergent,mathur2020graph}  considered in this subsection.  
In this context, the graph signals are the neural activity levels in regions of interest (ROIs) within the brain, 
and the graph nodes and edges represent the brain regions and the anatomical connection between them, respectively.
This application has been explored in  \ac{gsp} blind deconvolution studies \cite{segarra2016blind,ramirez2021graph}. 

In order to apply our graph-based sparse recovery methods to the problem of blind deconvolution, we propose the following two-stage approach. In the first stage, a blind deconvolution is performed to estimate the unknown graph filter coefficients, $\hat{h}_i$, $i=0,\ldots,\filterdeg$, from the input measurements,  
$\measurements$.
This estimation is implemented by one of the existing blind deconvolution methods, such as the methods in \cite{segarra2016blind, ramirez2021graph}.
The graph filter can then be recovered based on \eqref{eq; graph filter}, where the \ac{gso} associated with the underlying graph is assumed to be known. 
Then, in the second stage, our graph-based sparse recovery methods are employed to reconstruct the sparse graph signal $\signal$ based on the model in \eqref{eq; model}, where the graph filter used is the one estimated in the previous stage.  
This overall two-stage approach is illustrated in Fig. \ref{fig; blind methodology}. 

\begin{figure}[hbt]
    \centering
    \begin{tikzpicture}[node distance=4cm, auto]
        \tikzstyle{block} = [rectangle, draw, text centered, minimum height=2em, minimum width=2cm, outer sep=2pt]

        \node [block] (deconvolution) {\shortstack{Blind \\ Deconvolution}};
        \node [block, right of=deconvolution, xshift=0cm] (recovery) {\shortstack{Graph-Based \\ Sparse Recovery}};
        \node [right of=recovery, xshift=-2cm] (output) {$\hat{\signal}$};

        \draw[->, thick] (-2,0) -- ++(0.2,0) -- (deconvolution) node[midway, above] {$\measurements$};
        \draw[->, thick] (deconvolution) -- ++(1.2,0) -- (recovery) node[midway, above] {$\{\hat{h}_k\}_{k=1}^{\filterdeg}$};
        \draw[->, thick] (recovery) -- ++(1.3,0) -- (output);
    \end{tikzpicture}
    \caption{
    Two-stage approach for blind deconvolution combined with graph-based sparse recovery: The input measurements $\measurements$ undergo a blind deconvolution process to estimate the graph filter coefficients, $\{\hat{h}\}_k$, $k=1,\ldots,\filterdeg$. Then,  a graph-based sparse recovery is performed by the proposed methods,  resulting in the estimated signal $\hat{\signal}$.}
    \label{fig; blind methodology}
\end{figure}

\begin{figure*}[t]
        \centering
        \begin{minipage}[b]{0.3\textwidth}
            \centering
            \begin{subfigure}[b]{\textwidth}
                 \includegraphics[width=0.7\textwidth]{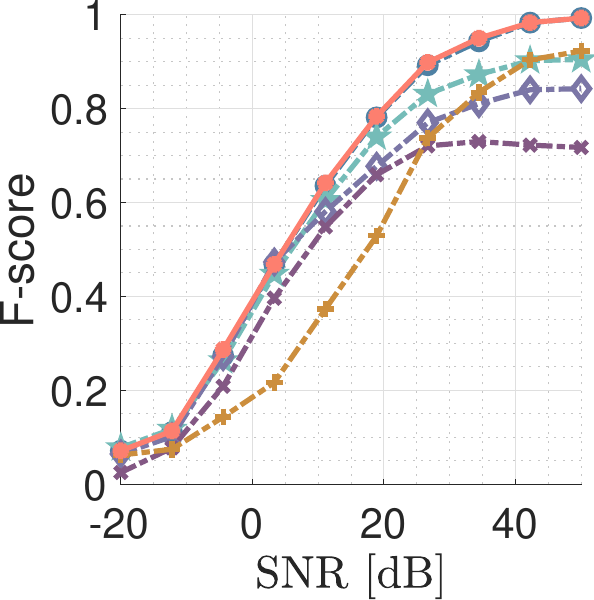}
                \caption*{\textbf{(a)} Known filter, F-score}
                \label{fig:sub11}
            \end{subfigure}
        \end{minipage}%
        \begin{minipage}[b]{0.3\textwidth}
            \centering
            \begin{subfigure}[b]{\textwidth}
               \includegraphics[width=0.7\textwidth]{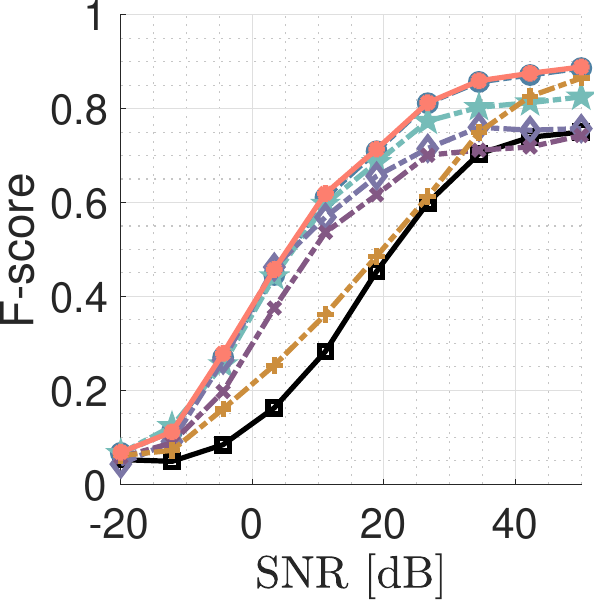}
                \caption*{\textbf{(b)} Blind Nuc. norm, F-score}
                \label{fig:sub21}
            \end{subfigure}
        \end{minipage}%
        \begin{minipage}[b]{0.3\textwidth}
            \centering
            \begin{subfigure}[b]{\textwidth}
                \includegraphics[width=0.7\textwidth]{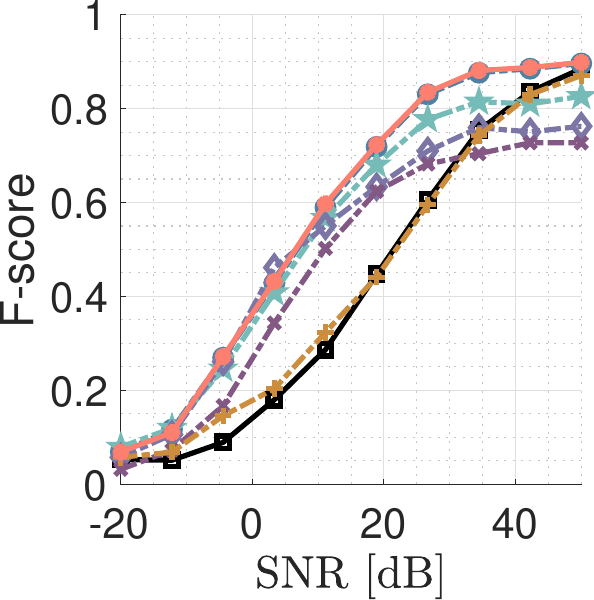}
                \caption*{\textbf{(c)} Blind Logdet, F-score}
                \label{fig:sub31}
            \end{subfigure}   
        \end{minipage}
\end{figure*}

    \begin{figure*}[t]
               \centering
        \begin{minipage}[b]{0.3\textwidth}
            \centering
            \begin{subfigure}[b]{\textwidth}
                 \includegraphics[width=0.7\textwidth]{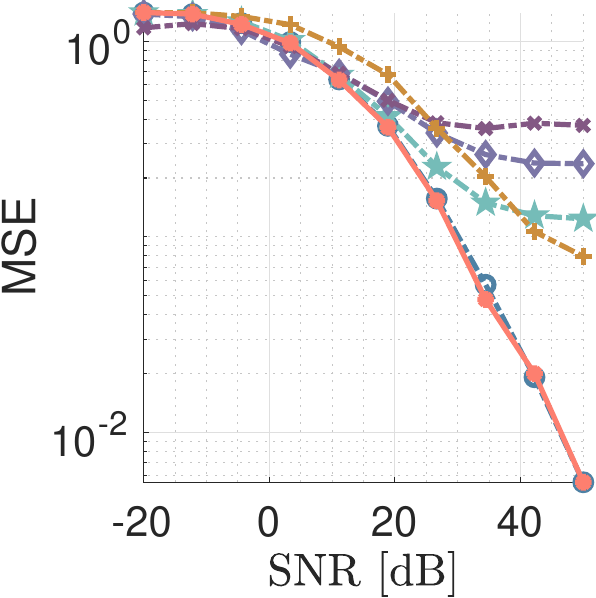}
                \caption*{\textbf{(d)} Known filter, \ac{mse}}
                \label{fig:sub12}
            \end{subfigure}
        \end{minipage}%
        \begin{minipage}[b]{0.3\textwidth}
            \centering
            \begin{subfigure}[b]{\textwidth}
               \includegraphics[width=0.7\textwidth]{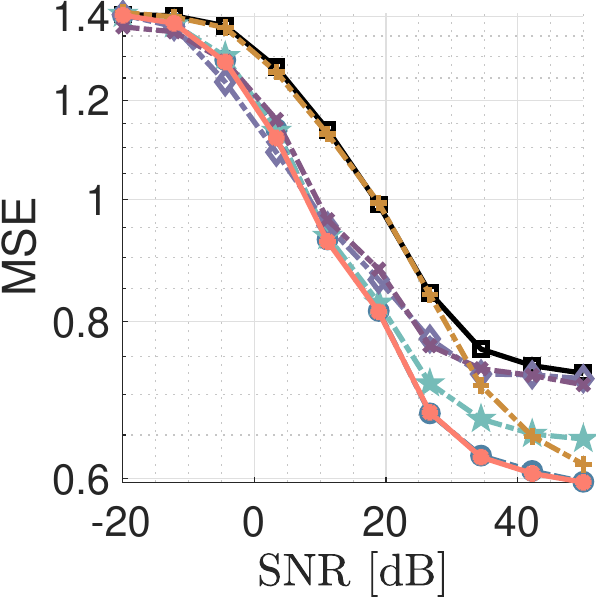}
                \caption*{\textbf{(e)} Blind Nuc. norm, \ac{mse}}
                \label{fig:sub22}
            \end{subfigure}
        \end{minipage}%
        \begin{minipage}[b]{0.3\textwidth}
            \centering
            \begin{subfigure}[b]{\textwidth}
                \includegraphics[width=0.7\textwidth]{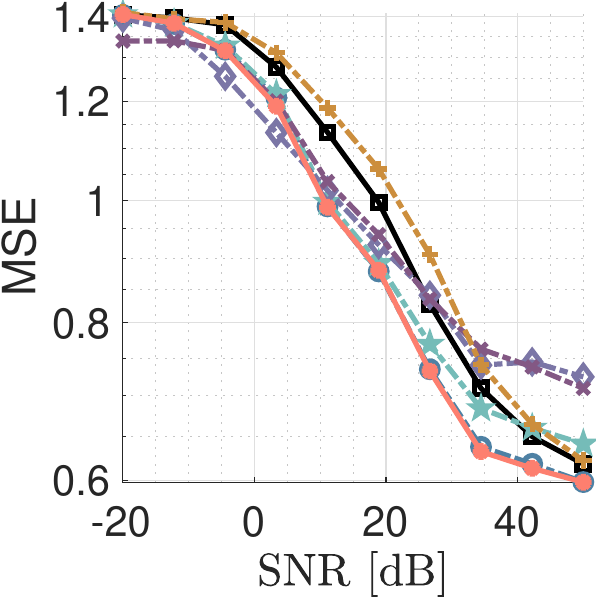}
                \caption*{\textbf{(f)} Blind Logdet, \ac{mse}}
                \label{fig:sub32}
            \end{subfigure}   
        \end{minipage}
        \vspace{0.3cm}

\begin{minipage}{0.9\textwidth}
        \centering
        \includegraphics[width=0.8\textwidth]{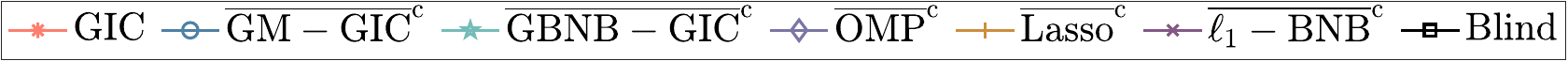}
        \label{fig:legend}
    \end{minipage}

   \caption{ Brain networks: F-score and \ac{mse} of graph-based sparse recovery methods using graph coefficients that are either known (see (a) and (d)) or estimated via the blind deconvolution techniques: Blind Nuc. norm ((b) and (e)) and Blind Logdet ((c) and (f)). } 
   \label{fig; blind}
\end{figure*}

In the following simulations, we consider the model in \eqref{eq; model}, where the graph filter is unknown 
and the underlying graph is a brain network. 
We use the Laplacian \ac{gso}, which is computed based on an adjacency matrix drawn from one of the six brain networks provided in 
\cite{iglesias2018demixing,iglesias_blind_separation}. 
 Entries in each row of the Laplacian that are significantly smaller (i.e. two orders of magnitude smaller than the maximum entry in the same row) are set to zero, and the Laplacian is subsequently normalized such that its largest eigenvalue is $12$ to establish a controlled setting for the analysis.

The graph filter is constructed by setting the graph filter degree to $\filterdeg = 3$, substituting the Laplacian matrix 
and the graph filter coefficients (which are uniformly drawn from 
the interval $[0,1]$) into \eqref{eq; graph filter}, and normalizing the result to achieve a Frobenius norm of $1$. 
The signal support is set to $\signalsupport = 4$ and is drawn according to Scenario $2$ from Subsection \ref{sec; sim; syntetic start}. 
The non-zero values of the graph signals are uniformly drawn from the interval $[0,1]$ and normalized to a unit-norm vector. 
This setup follows a similar framework to that considered in \cite{segarra2016blind}.

For the first stage of graph filter estimation, we use the following blind deconvolution techniques:
\begin{itemize}
    \item \textbf{Blind Nuc. norm:} Employs nuclear norm minimization combined with  $\ell_1$ row sparsity regularization
    (see Equation (9) in \cite{segarra2016blind}).
    \item \textbf{Blind Logdet:} Employs a log surrogate combined with $\ell_2$ row sparsity regularization (see Equation (17) with $\tau_x=0.5$ and $\tau_h=0$ in \cite{ramirez2021graph}). 
\end{itemize}
The performance of each scenario is evaluated through at least $300$ Monte Carlo simulations.

We first evaluate the existence of Assumption \ref{A; graph degree}, which assumes a low maximal degree for the graph, 
in the context of the real-data brain networks. This assumption inherently indicates that the graph is sparse. 
  Table \ref{table; brain} presents 
the ratio of the number of edges to the number of possible connections, $\frac{\card{\edges}}{\frac{1}{2}\card{\nodes}(\card{\nodes}-1)}$, and the maximal degree, $d_{\text{max}}$, 
for the six real-data brain networks derived from the dataset in \cite{hagmann2008mapping}.
as described above. 
This table demonstrates that these brain networks are sparse graphs with low maximal degrees, thus aligning with Assumption \ref{A; graph degree}.
\begin{table}[hbt]

\centering
\begin{tabular}{c| c c c c c c}
\hline
& $\mathcal{G}_1$ & $\mathcal{G}_2$ & $\mathcal{G}_3$ & $\mathcal{G}_4$ & $\mathcal{G}_5$ & $\mathcal{G}_6$ \\ 
\hline
$\frac{\card{\text{edges}}}{\frac{1}{2}\card{\nodes}(\card{\nodes}-1)}$   & 0.1 & 0.11 & 0.09 & 0.11 & 0.11 & 0.11  \\ 
$d_{\text{max}}$       & 3.28    & 4.41    & 5.83    & 3.87    & 3.93    & 3.68    \\ 
\bottomrule
\end{tabular}
\caption{ Brain networks: the ratio of the number of edges and the number of possible connections, $\frac{\card{\text{edges}}}{\card{K_n}}$, and the maximal degree, $d_{\text{max}}$, 
for the real-data brain networks from \cite{hagmann2008mapping,iglesias_blind_separation}.
}
\label{table; brain}
\end{table}

Next, we evaluate the graph-based methods from  Subsection \ref{sec; sim; methods and measures}
for three scenarios: 1) using a known graph filter; and 2) and 3) using the output of the two-stage blind deconvolution approach with the two different techniques, referred to as ``Blind Nuc. norm" and ``Blind Logdet", respectively.
We compare the result of the different methods with those of the existing blind deconvolution methods, which can be interpreted as performing only Stage 1 of the approach described in Fig. \ref{fig; blind methodology}. 
These methods are referred to as ``Blind" in Fig. \ref{fig; blind}. 
In particular, in Figs. \ref{fig; blind methodology}.(b) and \ref{fig; blind methodology}.(e), the Blind Nuc. norm method is used, while in Figs. \ref{fig; blind methodology}.(c) and \ref{fig; blind methodology}.(f) the BLind Logdet method is employed.
For the clarity of presentation, we only present the results for the graph-based methods after applying the \ac{gfoc} method. This is because the \ac{gfoc} method consistently improves the recovery performance with negligible additional computational overhead (run-time).

Figure \ref{fig; blind} presents the F-score and \ac{mse} performance of the different methods as a function of the \ac{snr}. 
The \ac{mse} 
is computed as
$
MSE = \frac{1}{\norm{\signal\hvec^T }}\norm{\signal\hvec^T - \hat{\signal}\hat{\hvec}^T}_F,
$
where $\signal$ is the sparse graph signal, $\hat{\signal}$ is its estimator, $\hvec$ represents the vector of graph filter coefficients, and $\hat{\hvec}$ is its estimator. 
For a known graph filter, 
we use $\hat{\hvec}=\hvec$.

Figures \ref{fig; blind}.(b)-(c) and \ref{fig; blind}.(e)-(f) show that the \ac{gic}, $\overline{\ac{gmgic}}^{c}$, $\overline{\text{\gbnb}}^{c}$, and $\overline{\ac{omp}}^{c}$ methods significantly outperform the blind deconvolution techniques in both F-score and \ac{mse} under the case where the graph filter is unknown. 
These results confirm that the proposed sparse recovery methods are also applicable in scenarios where the graph filter is not known {\em{a priori}}.
 In addition, it can be seen in these figures that while $\overline{\ac{omp}}^{c}$ achieves competitive results at low \ac{snr} levels, its performance degrades as the \ac{snr} increases, 
 and is worse than those of the $\overline{\ac{gmgic}}^{c}$ and $\overline{\text{\gbnb}}^{c}$ methods.
The $\overline{\text{Lasso}}^{c}$, is competitive to the above methods only at high \ac{snr} values, 
and the $\overline{\ell_1\text{-}BnB}^{c}$ method works poorly for all \ac{snr} rates and even downgrades 
the blind deconvolution output. 
Notably, the F-scores of the graph-based sparse recovery methods, using the two-stage approach from Fig. \ref{fig; blind methodology}, are comparable to those obtained where the graph filter is known, as can be seen by comparing the results in Figs. \ref{fig; blind}(a) with the results in Figs. \ref{fig; blind}. (b)-(c). 
In contrast, the \ac{mse} is more sensitive to the accuracy of the graph filter coefficients, indicating that improving the filter estimation could lead to better \ac{mse} results.
It is also important to emphasize that the additional run-time for the $\overline{\text{\gbnb}}^{c}$, $\overline{\ac{omp}}^{c}$, $\overline{\ell_1\text{-}BnB}^{c}$, and $\overline{\text{Lasso}}^{c}$ methods have been found to be negligible compared to the computational cost of the first stage of the blind deconvolution.
Thus, it is recommended that a subsequent graph sparse recovery stage be employed on the existing graph blind deconvolution methods.

\label{sec; brain end}
\subsection{Test Case C - Cholera Outbreak} \label{sec; cholera start}
In this section, we evaluate the performance of the different methods in a real-world test case where both the graph (network) and the graph filter are unknown. While the observation model in \eqref{eq; model} forms the basis for our approach and  underpins the proposed sparse recovery methods,  it is important to note that in this scenario, the true structure of the real-world data and the underlying model are unknown. 
Thus, the simulations in this subsection demonstrate how the proposed framework
can be applied to 
the practical application of interpreting 
real-world data,  
despite the uncertainty in the true graph structure and graph filter that govern the system dynamics.

The real-world data for this case study, drawn from \cite{wilson2012cholera}, 
involves the 1854 Cholera outbreak in London's Soho district,
 where Dr. John Snow identified the source of the epidemic as a contaminated water pump \cite{snow1854report}. 
 Inspired by the framework in \cite{pena2016source}, this study applies the proposed graph-based sparse recovery methods to identify the source of the outbreak using historical data. The used dataset  includes information on death tolls across 250 buildings, their  geographic locations, and the positions of 8 water pumps, as provided in \cite{wilson2012cholera}.

First, we model the underlying graph using a type-aware nearest neighbor structure. The graph is constructed based on the hypothesis that water was the primary  transmission vector for cholera, with nodes representing buildings and pumps. For the graph edges, there is an edge between each pump and its two nearest pumps, each pump and its five nearest buildings, and each building and its two nearest buildings. 
The edge weights are computed using a Gaussian kernel:
$$
W_{\none,\ntwo } = \exp\left(-\frac{(\text{dist}(\none, \ntwo))^2}{2 \cdot (f(\none,\ntwo))^2}\right),
$$
where dist represents the Euclidean distance, and  $f(i,j)$ depends on the node types: \( f(\none,\ntwo) = 1000 \) for pump connections, \( f(\none,\ntwo) = 80 \) for pump-building connections, and \( f(\none,\ntwo) = 20 \) for buildings connections.
 For ground truth, we consider the infected pump and the building with the highest death toll. To promote sparsity, we reduce the death count by one in each building, resulting in $136$ non-zero elements out of $258$, and we normalize the data vector to $1$. 
It should be noted that the underlying graph and the graph filter are set to  satisfy Assumptions \ref{A; graph degree}-\ref{A; filter degree}. 
Using the reconstructed graph, we apply the methodology outlined in Fig. \ref{fig; blind methodology}, which includes a blind deconvolution step followed by graph-based sparse recovery. This follows the approach detailed in Subsection \ref{sec; brain start}, with the estimated graph. Accordingly, we use the same methods as those described in Subsection \ref{sec; brain start}, but with the estimated graph structure.
The performance is evaluated using at least $60$ Monte Carlo simulations.

\begin{figure}[t]
\centering
  \includegraphics[width=5.5cm]{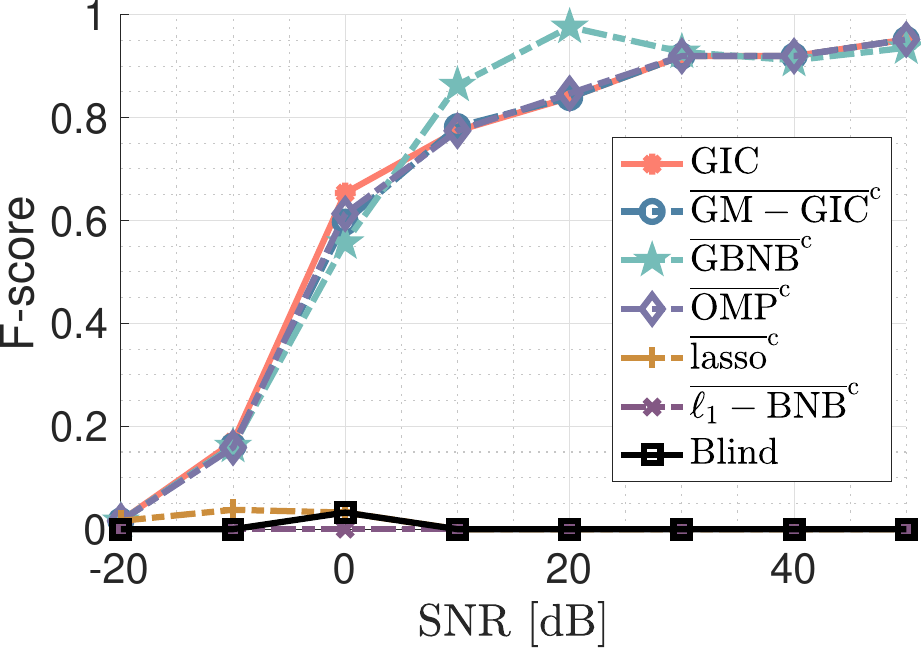}
\caption{Cholera network: the F-score of graph-based sparse recovery methods using coefficients estimated via the Bilnd Nuc. norm method.
} 
\label{fig; snow results} 
\end{figure} 

Figure \ref{fig; snow results} shows the F-scores of the graph-based sparse recovery methods: \ac{gic}, $\overline{\ac{gmgic}}^{c}$, $\overline{\text{\gbnb}}^{c}$, $\overline{\ac{omp}}^{c}$, $\overline{\ell_1\text{-}BnB}^{c}$, and $\overline{\text{Lasso}}^{c}$, as a function of \ac{snr}.
In each simulation, we added a zero-mean Gaussian additive noise to the data vector to obtain the specific \ac{snr}.
This evaluation focuses on the challenge of identifying the ground truth provided by the data vector contaminated with noise 
under different \acp{snr}.
The results indicate that the 
blind deconvolution method Blind Nuc. norm  \cite{segarra2016blind} cannot reconstruct the support set at any \ac{snr} level in this case.
In contrast, 
the \ac{gic}, $\overline{\ac{gmgic}}^{c}$, $\overline{\text{\gbnb}}^{c}$, and $\overline{\ac{omp}}^{c}$ methods have a significantly higher F-score and achieve close to perfect reconstruction (F-score$>0.9$) for large \ac{snr}.
These methods achieve similar F-scores for most \acp{snr}, where the F-score increases as the \ac{snr} increases. 
Between $15$ [dB] and $25$ [dB], $\overline{\text{\gbnb}}^{c}$ achieves slightly higher results. 
Thus, the results in Fig. \ref{fig; snow results} demonstrate that the proposed graph-based sparse recovery methods offer substantial improvements in support set recovery, even in challenging real-world scenarios with unknown graph structures and filters, outperforming the existing blind deconvolution techniques.
Notably, the run-time added by the graph-based methods, excluding the \ac{gic} method, 
is negligible in comparison to the run-time of the blind deconvolution method.


\label{sec; cholera end}
\section{Conclusions}  \label{sec; conclusions}

In this paper, we propose three sparse recovery techniques: the \ac{gmgic}, \ac{gbnb}, and \ac{gfoc} methods. 
These techniques solve the support set recovery problem for diffused sparse graph signals,  which are node-domain sparse graph signals diffused over the graph by a graph filter.
The proposed methods use the \ac{gic} as the optimization cost function and leverage the underlying graphical structure to reduce the computational complexity and enhance the recovery performance. In particular, the \ac{gmgic} and \ac{gbnb} directly tackle the support set recovery problem by considering local graph properties and efficiently partitioning the dictionary elements 
and iteratively searching over candidate support sets, respectively.
The \ac{gfoc} method is based on correcting the solution obtained by any other support set recovery technique.
Furthermore, we provide a theoretical analysis of the considered setting with: 1) an examination of the dictionary matrix atoms (i.e. graph filter columns) derived from the underlying graphical structure; and 2) an evaluation of the computational complexity of the proposed methods, based on
the graph filter degree, the maximum degree of the graph, and sparsity restrictions.

Simulations conducted over the  \ac{sbm} graphs  reveal that the proposed \ac{gmgic} and \ac{gbnb} methods achieve higher support set recovery  and signal recovery accuracy  compared to state-of-the-art methods such as \ac{omp}, Lasso, and $\ell_1$-\ac{bnb}, while maintaining significantly lower computational overhead than the optimal \ac{gic} method.
In general, the \ac{gmgic} method outperforms the \ac{gbnb} method in terms of the accuracy of the support set recovery, but requires more computing time. 
It is also demonstrated that
the \ac{gfoc} method can significantly improve the support set recovery accuracy of existing methods, including the \ac{omp}, Lasso, and $\ell_1$-\ac{bnb} methods, without imposing significant additional computation overhead. 
By applying the \ac{gfoc} method on the \ac{gbnb} method, we obtained support set estimation results comparable to the \ac{gic} benchmark, with a run-time comparable to the \ac{omp} and Lasso methods.
Consequently, the best performance is obtained by combining the \ac{gbnb} and \ac{gfoc} methods. 
In addition, we show the application of these methods for the problem of blind deconvolution in which the graph signal and the graph filter  are unknown. \label{pp; conclusion blind 1}
Simulations conducted on brain networks and real data on the Cholera $1854$ outbreak show that applying the proposed methods significantly improves the signal recovery performance for the application of blind deconvolution. \label{pp; conclusion blind 2}


\bibliographystyle{IEEEtran}
\bibliography{c_auxiliary/refs}

\begin{thebibliography}{10}
\providecommand{\url}[1]{#1}
\csname url@samestyle\endcsname
\providecommand{\newblock}{\relax}
\providecommand{\bibinfo}[2]{#2}
\providecommand{\BIBentrySTDinterwordspacing}{\spaceskip=0pt\relax}
\providecommand{\BIBentryALTinterwordstretchfactor}{4}
\providecommand{\BIBentryALTinterwordspacing}{\spaceskip=\fontdimen2\font plus
\BIBentryALTinterwordstretchfactor\fontdimen3\font minus
  \fontdimen4\font\relax}
\providecommand{\BIBforeignlanguage}[2]{{%
\expandafter\ifx\csname l@#1\endcsname\relax
\typeout{** WARNING: IEEEtran.bst: No hyphenation pattern has been}%
\typeout{** loaded for the language `#1'. Using the pattern for}%
\typeout{** the default language instead.}%
\else
\language=\csname l@#1\endcsname
\fi
#2}}
\providecommand{\BIBdecl}{\relax}
\BIBdecl

\bibitem{morgenstern2023sparse}
G.~Morgenstern and T.~Routtenberg, ``Sparse graph signal recovery by the
  graph-based multiple generalized information criterion ({GM-GIC}),'' in
  \emph{International Workshop on Computational Advances in Multi-Sensor
  Adaptive Processing (CAMSAP)}, 2023, pp. 491--495.

\bibitem{ortega2018graph}
A.~Ortega, P.~Frossard, J.~Kova{\v{c}}evi{\'c}, J.~M. Moura, and
  P.~Vandergheynst, ``Graph signal processing: Overview, challenges, and
  applications,'' \emph{Proc. IEEE}, vol. 106, no.~5, pp. 808--828, 2018.

\bibitem{shuman2013emerging}
D.~I. Shuman, S.~K. Narang, P.~Frossard, A.~Ortega, and P.~Vandergheynst, ``The
  emerging field of signal processing on graphs: Extending high-dimensional
  data analysis to networks and other irregular domains,'' \emph{IEEE Signal
  Process. Mag.}, vol.~30, no.~3, pp. 83--98, 2013.

\bibitem{sandryhaila2013discrete}
A.~Sandryhaila and J.~M. Moura, ``Discrete signal processing on graphs,''
  \emph{IEEE Trans. Signal Process.}, vol.~61, no.~7, pp. 1644--1656, 2013.

\bibitem{sandryhaila2014discrete}
------, ``Discrete signal processing on graphs: Frequency analysis,''
  \emph{IEEE Trans. Signal Process.}, vol.~62, no.~12, pp. 3042--3054, 2014.

\bibitem{marques2017stationary}
A.~G. Marques, S.~Segarra, G.~Leus, and A.~Ribeiro, ``Stationary graph
  processes and spectral estimation,'' \emph{IEEE Trans. Signal Process.},
  vol.~65, no.~22, pp. 5911--5926, 2017.

\bibitem{drayer2019detection}
E.~Drayer and T.~Routtenberg, ``Detection of false data injection attacks in
  smart grids based on graph signal processing,'' \emph{IEEE Syst. J.},
  vol.~14, no.~2, pp. 1886--1896, 2019.

\bibitem{tanaka2020sampling}
Y.~Tanaka, Y.~C. Eldar, A.~Ortega, and G.~Cheung, ``Sampling signals on graphs:
  From theory to applications,'' \emph{IEEE Signal Process. Mag.}, vol.~37,
  no.~6, pp. 14--30, 2020.

\bibitem{marques2015sampling}
A.~G. Marques, S.~Segarra, G.~Leus, and A.~Ribeiro, ``Sampling of graph signals
  with successive local aggregations,'' \emph{IEEE Trans. Signal Process.},
  vol.~64, no.~7, pp. 1832--1843, 2015.

\bibitem{chen2015signal}
S.~Chen, A.~Sandryhaila, J.~M. Moura, and J.~Kova{\v{c}}evi{\'c}, ``Signal
  recovery on graphs: Variation minimization,'' \emph{IEEE Trans. Signal
  Process.}, vol.~63, no.~17, pp. 4609--4624, 2015.

\bibitem{di2018adaptive}
P.~Di~Lorenzo, P.~Banelli, E.~Isufi, S.~Barbarossa, and G.~Leus, ``Adaptive
  graph signal processing: Algorithms and optimal sampling strategies,''
  \emph{IEEE Trans. Signal Process.}, vol.~66, no.~13, pp. 3584--3598, 2018.

\bibitem{sagi2022gsp}
G.~Sagi and T.~Routtenberg, ``{MAP} estimation of graph signals,'' \emph{IEEE
  Trans. Signal Process.}, vol.~72, pp. 463--479, 2024.

\bibitem{kroizer2022bayesian}
A.~Kroizer, T.~Routtenberg, and Y.~C. Eldar, ``Bayesian estimation of graph
  signals,'' \emph{IEEE Trans. Signal Process.}, vol.~70, pp. 2207--2223, 2022.

\bibitem{amar2023widely}
A.~Amar and T.~Routtenberg, ``Widely-linear mmse estimation of complex-valued
  graph signals,'' \emph{IEEE Trans. Signal Process.}, vol.~71, pp. 1770--1785,
  2023.

\bibitem{dabush2023verifying}
L.~Dabush and T.~Routtenberg, ``Verifying the smoothness of graph signals: A
  graph signal processing approach,'' \emph{\normalfont{arXiv preprint:
  2307.03210}}, 2023.

\bibitem{venkitaraman2019predicting}
A.~Venkitaraman, S.~Chatterjee, and P.~H{\"a}ndel, ``Predicting graph signals
  using kernel regression where the input signal is agnostic to a graph,''
  \emph{IEEE Trans. Signal Inf. Process. Netw.}, vol.~5, no.~4, pp. 698--710,
  2019.

\bibitem{isufi2016autoregressive}
E.~Isufi, A.~Loukas, A.~Simonetto, and G.~Leus, ``Autoregressive moving average
  graph filtering,'' \emph{IEEE Trans. Signal Process.}, vol.~65, no.~2, pp.
  274--288, 2016.

\bibitem{coutino2019advances}
M.~Coutino, E.~Isufi, and G.~Leus, ``Advances in distributed graph filtering,''
  \emph{IEEE Trans. Signal Process.}, vol.~67, no.~9, pp. 2320--2333, 2019.

\bibitem{Chen_Donoho_2001}
S.~S. Chen, D.~L. Donoho, and M.~A. Saunders, ``Atomic decomposition by basis
  pursuit,'' \emph{SIAM Review}, vol.~43, no.~1, pp. 129--159, 2001.

\bibitem{Tropp_2004}
J.~A. Tropp, ``Greed is good: algorithmic results for sparse approximation,''
  \emph{IEEE Trans. Inf. Theory}, vol.~50, no.~10, pp. 2231--2242, Oct. 2004.

\bibitem{donoho2005stable}
D.~L. Donoho, M.~Elad, and V.~N. Temlyakov, ``Stable recovery of sparse
  overcomplete representations in the presence of noise,'' \emph{IEEE Trans.
  Inf. Theory}, vol.~52, no.~1, pp. 6--18, 2005.

\bibitem{elad2010sparse}
M.~Elad, \emph{Sparse and Redundant Representations: From Theory to
  Applications in Signal and Image Processing}.\hskip 1em plus 0.5em minus
  0.4em\relax New York, NY, USA: Springer Science \& Business Media, 2010.

\bibitem{tibshirani1996regression}
R.~Tibshirani, ``Regression shrinkage and selection via the lasso,''
  \emph{Journal of the Royal Statistical Society Series B: Statistical
  Methodology}, vol.~58, no.~1, pp. 267--288, 1996.

\bibitem{efron2004least}
B.~Efron, T.~Hastie, I.~Johnstone, and R.~Tibshirani, ``Least angle
  regression,'' \emph{The Annals of Statistics}, vol.~32, no.~2, pp. 407 --
  499, 2004.

\bibitem{candes2007dantzig}
E.~Candes and T.~Tao, ``The dantzig selector: Statistical estimation when p is
  much larger than n,'' \emph{The Annals of Statistics}, vol.~35, no.~6, pp.
  2313 -- 2351, 2007.

\bibitem{donoho2009message}
D.~L. Donoho, A.~Maleki, and A.~Montanari, ``Message-passing algorithms for
  compressed sensing,'' \emph{Proceedings of the National Academy of Sciences},
  vol. 106, no.~45, pp. 18\,914--18\,919, 2009.

\bibitem{wipf2004sparse}
D.~P. Wipf and B.~D. Rao, ``Sparse bayesian learning for basis selection,''
  \emph{IEEE Trans. Signal Process.}, vol.~52, no.~8, pp. 2153--2164, 2004.

\bibitem{gorodnitsky1997sparse}
I.~F. Gorodnitsky and B.~D. Rao, ``Sparse signal reconstruction from limited
  data using focuss: A re-weighted minimum norm algorithm,'' \emph{IEEE Trans.
  Signal Process.}, vol.~45, no.~3, pp. 600--616, 1997.

\bibitem{chartrand2008iteratively}
R.~Chartrand and W.~Yin, ``Iteratively reweighted algorithms for compressive
  sensing,'' in \emph{ICASSP}, 2008, pp. 3869--3872.

\bibitem{mallat1993matching}
S.~G. Mallat and Z.~Zhang, ``Matching pursuits with time-frequency
  dictionaries,'' \emph{IEEE Trans. Signal Process.}, vol.~41, no.~12, pp.
  3397--3415, 1993.

\bibitem{cai2011orthogonal}
T.~T. Cai and L.~Wang, ``Orthogonal matching pursuit for sparse signal recovery
  with noise,'' \emph{IEEE Trans. Inf. Theory}, vol.~57, no.~7, pp. 4680--4688,
  2011.

\bibitem{dai2009subspace}
W.~Dai and O.~Milenkovic, ``Subspace pursuit for compressive sensing signal
  reconstruction,'' \emph{IEEE Trans. Inf. Theory}, vol.~55, no.~5, pp.
  2230--2249, 2009.

\bibitem{blumensath2008gradient}
T.~Blumensath and M.~E. Davies, ``Gradient pursuits,'' \emph{IEEE Trans. Signal
  Process.}, vol.~56, no.~6, pp. 2370--2382, 2008.

\bibitem{haupt2008compressed}
J.~Haupt, W.~U. Bajwa, M.~Rabbat, and R.~Nowak, ``Compressed sensing for
  networked data,'' \emph{IEEE Signal Process. Mag.}, vol.~25, no.~2, pp.
  92--101, 2008.

\bibitem{xu2011compressive}
W.~Xu, E.~Mallada, and A.~Tang, ``Compressive sensing over graphs,'' in
  \emph{2011 Proceedings IEEE INFOCOM}.\hskip 1em plus 0.5em minus 0.4em\relax
  IEEE, 2011, pp. 2087--2095.

\bibitem{cheraghchi2012graph}
M.~Cheraghchi, A.~Karbasi, S.~Mohajer, and V.~Saligrama, ``Graph-constrained
  group testing,'' \emph{IEEE Trans. Inf. Theory}, vol.~58, no.~1, pp.
  248--262, 2012.

\bibitem{pinto2012locating}
P.~C. Pinto, P.~Thiran, and M.~Vetterli, ``Locating the source of diffusion in
  large-scale networks,'' \emph{Physical review letters}, vol. 109, no.~6, p.
  068702, 2012.

\bibitem{sefer2016diffusion}
E.~Sefer and C.~Kingsford, ``Diffusion archeology for diffusion progression
  history reconstruction,'' \emph{Knowledge and information systems}, vol.~49,
  pp. 403--427, 2016.

\bibitem{zhang2016towards}
P.~Zhang, J.~He, G.~Long, G.~Huang, and C.~Zhang, ``Towards anomalous diffusion
  sources detection in a large network,'' \emph{ACM Transactions on Internet
  Technology (TOIT)}, vol.~16, no.~1, pp. 1--24, 2016.

\bibitem{pena2016source}
R.~Pena, X.~Bresson, and P.~Vandergheynst, ``Source localization on graphs via
  $\ell$1 recovery and spectral graph theory,'' in \emph{Image, Video, and
  Multidimensional Signal Processing Workshop (IVMSP)}, 2016, pp. 1--5.

\bibitem{sridhar2023quickest}
A.~Sridhar, T.~Routtenberg, and H.~V. Poor, ``Quickest inference of
  susceptible-infected cascades in sparse networks,'' in \emph{International
  Symposium on Information Theory (ISIT)}, 2023, pp. 102--107.

\bibitem{wang2015local}
X.~Wang, P.~Liu, and Y.~Gu, ``Local-set-based graph signal reconstruction,''
  \emph{IEEE Trans. Signal Process.}, vol.~63, no.~9, pp. 2432--2444, 2015.

\bibitem{anis2014towards}
A.~Anis, A.~Gadde, and A.~Ortega, ``Towards a sampling theorem for signals on
  arbitrary graphs,'' in \emph{ICASSP}, 2014, pp. 3864--3868.

\bibitem{chen2015discrete}
S.~Chen, R.~Varma, A.~Sandryhaila, and J.~Kovacevic, ``Discrete signal
  processing on graphs: Sampling theory,'' \emph{IEEE Trans. Signal Process.},
  vol.~63, no.~24, pp. 6510--6523, 2015.

\bibitem{tsitsvero2016signals}
M.~Tsitsvero, S.~Barbarossa, and P.~Di~Lorenzo, ``Signals on graphs:
  Uncertainty principle and sampling,'' \emph{IEEE Trans. Signal Process.},
  vol.~64, no.~18, pp. 4845--4860, 2016.

\bibitem{mashhadi2017interpolation}
M.~B. Mashhadi, M.~Fallah, and F.~Marvasti, ``Interpolation of sparse graph
  signals by sequential adaptive thresholds,'' in \emph{2017 International
  Conference on Sampling Theory and Applications (SampTA)}.\hskip 1em plus
  0.5em minus 0.4em\relax IEEE, 2017, pp. 266--270.

\bibitem{romero2016kernel}
D.~Romero, M.~Ma, and G.~B. Giannakis, ``Kernel-based reconstruction of graph
  signals,'' \emph{IEEE Transactions on Signal Processing}, vol.~65, no.~3, pp.
  764--778, 2016.

\bibitem{ramirez2021graph}
D.~Ram{\'\i}rez, A.~G. Marques, and S.~Segarra, ``Graph-signal reconstruction
  and blind deconvolution for structured inputs,'' \emph{Signal Process.}, vol.
  188, p. 108180, 2021.

\bibitem{segarra2016blind}
S.~Segarra, G.~Mateos, A.~G. Marques, and A.~Ribeiro, ``Blind identification of
  graph filters,'' \emph{IEEE Trans. Signal Process.}, vol.~65, no.~5, pp.
  1146--1159, 2016.

\bibitem{segarra2017optimal}
S.~Segarra, A.~G. Marques, and A.~Ribeiro, ``Optimal graph-filter design and
  applications to distributed linear network operators,'' \emph{IEEE Trans.
  Signal Process.}, vol.~65, no.~15, pp. 4117--4131, 2017.

\bibitem{segarra2015distributed}
------, ``Distributed implementation of linear network operators using graph
  filters,'' in \emph{Annual Allerton Conference on Communication, Control, and
  Computing (Allerton)}, 2015, pp. 1406--1413.

\bibitem{mei2015signal}
J.~Mei and J.~M. Moura, ``Signal processing on graphs: Estimating the structure
  of a graph,'' in \emph{ICASSP}, 2015, pp. 5495--5499.

\bibitem{rainer2002opinion}
H.~Rainer and U.~Krause, ``Opinion dynamics and bounded confidence: Models,
  analysis and simulation,'' \emph{Journal of Artificial Societies and Social
  Simulation}, vol.~5, no.~3, 2002.

\bibitem{shah2011rumors}
D.~Shah and T.~Zaman, ``Rumors in a network: Who's the culprit?'' \emph{IEEE
  Trans. Inf. Theory}, vol.~57, no.~8, pp. 5163--5181, 2011.

\bibitem{newman2002spread}
M.~E. Newman, ``Spread of epidemic disease on networks,'' \emph{Physical review
  E}, vol.~66, no.~1, p. 016128, 2002.

\bibitem{shah2010detecting}
D.~Shah and T.~Zaman, ``Detecting sources of computer viruses in networks:
  theory and experiment,'' in \emph{Proceedings of the ACM SIGMETRICS
  international conference on Measurement and modeling of computer systems},
  2010, pp. 203--214.

\bibitem{morgenstern2020structuralconstrained}
G.~{Morgenstern} and T.~{Routtenberg}, ``Structural-constrained methods for the
  identification of unobservable false data injection attacks in power
  systems,'' \emph{IEEE Access}, vol.~10, pp. 94\,169--94\,185, 2022.

\bibitem{morgenstern2023protection}
G.~Morgenstern, J.~Kim, J.~Anderson, G.~Zussman, and T.~Routtenberg,
  ``Protection against graph-based false data injection attacks on power
  systems,'' \emph{IEEE Trans. Control Netw. Syst.}, pp. 1--12, 2024.

\bibitem{segarra2016reconstruction}
S.~Segarra, A.~G. Marques, G.~Leus, and A.~Ribeiro, ``Reconstruction of graph
  signals through percolation from seeding nodes,'' \emph{IEEE Trans. Signal
  Process.}, vol.~64, no.~16, pp. 4363--4378, 2016.

\bibitem{zhu2020estimating}
Y.~Zhu, F.~J.~I. Garcia, A.~G. Marques, and S.~Segarra, ``Estimating network
  processes via blind identification of multiple graph filters,'' \emph{IEEE
  Trans. Signal Process.}, vol.~68, pp. 3049--3063, 2020.

\bibitem{ye2018blind}
C.~Ye, R.~Shafipour, and G.~Mateos, ``Blind identification of invertible graph
  filters with multiple sparse inputs,'' in \emph{European Signal Processing
  Conference (EUSIPCO)}, 2018, pp. 121--125.

\bibitem{iglesias2018demixing}
F.~J. Iglesias, S.~Segarra, S.~Rey-Escudero, A.~G. Marques, and
  D.~Ram{\'\i}rez, ``Demixing and blind deconvolution of graph-diffused sparse
  signals,'' in \emph{ICASSP}, 2018, pp. 4189--4193.

\bibitem{rey2019sampling}
S.~Rey-Escudero, F.~J.~I. Garcia, C.~Cabrera, and A.~G. Marques, ``Sampling and
  reconstruction of diffused sparse graph signals from successive local
  aggregations,'' \emph{IEEE Signal Process. Lett.}, vol.~26, no.~8, pp.
  1142--1146, 2019.

\bibitem{stoica2004model}
P.~Stoica and Y.~Selen, ``Model-order selection: a review of information
  criterion rules,'' \emph{IEEE Signal Process. Mag.}, vol.~21, no.~4, pp.
  36--47, 2004.

\bibitem{boyd2007branch}
S.~Boyd and J.~Mattingley, ``Branch and bound methods,'' \emph{Notes for
  EE364b, Stanford University}, vol. 2006, p.~07, 2007.

\bibitem{lawler1966branch}
E.~L. Lawler and D.~E. Wood, ``Branch-and-bound methods: A survey,''
  \emph{Operations research}, vol.~14, no.~4, pp. 699--719, 1966.

\bibitem{hagmann2008mapping}
P.~Hagmann, L.~Cammoun, X.~Gigandet, R.~Meuli, C.~J. Honey, V.~J. Wedeen, and
  O.~Sporns, ``Mapping the structural core of human cerebral cortex,''
  \emph{PLoS biology}, vol.~6, no.~7, p. e159, 2008.

\bibitem{wilson2012cholera}
R.~Wilson, ``John {S}now's famous {C}holera analysis data in modern {GIS}
  formats,'' 2012.

\bibitem{snow1854report}
J.~Snow \emph{et~al.}, ``Report on the cholera outbreak in the parish of st.
  james,'' \emph{James, Westminster, during the Autumn of 1854}, pp. 97--120,
  July 1855.

\bibitem{wang2010generating}
Z.~Wang, A.~Scaglione, and R.~J. Thomas, ``Generating statistically correct
  random topologies for testing smart grid communication and control
  networks,'' \emph{IEEE Trans. Smart Grid}, vol.~1, pp. 28--39, 2010.

\bibitem{kay1993fundamentals}
S.~M. Kay, \emph{Fundamentals of statistical signal processing: Estimation
  Theory}.\hskip 1em plus 0.5em minus 0.4em\relax New Jersey, NJ, USA: Prentice
  Hall PTR, 1993, vol.~1.

\bibitem{yanai2011projection}
H.~Yanai, K.~Takeuchi, and Y.~Takane, ``Projection matrices,'' in
  \emph{Projection Matrices, Generalized Inverse Matrices, and Singular Value
  Decomposition}.\hskip 1em plus 0.5em minus 0.4em\relax Springer, 2011, pp.
  25--54.

\bibitem{babu2023multiple}
P.~Babu and P.~Stoica, ``Multiple-hypothesis testing rules for high-dimensional
  model selection and sparse-parameter estimation,'' \emph{Signal Process.},
  vol. 213, p. 109189, 2023.

\bibitem{kayfundamentals}
S.~M. Kay, \emph{Fundamentals of Statistical Signal Processing, Volume II:
  Detection Theory}.\hskip 1em plus 0.5em minus 0.4em\relax New Jersey, NJ,
  USA: Prentice Hall PTR, 1993.

\bibitem{dionne1979exact}
R.~Dionne and M.~Florian, ``Exact and approximate algorithms for optimal
  network design,'' \emph{Networks}, vol.~9, no.~1, pp. 37--59, 1979.

\bibitem{solanki1998using}
R.~S. Solanki, J.~K. Gorti, and F.~Southworth, ``Using decomposition in
  large-scale highway network design with a quasi-optimization heuristic,''
  \emph{Transportation Research Part B: Methodological}, pp. 127--140, 1998.

\bibitem{sokolova}
M.~Sokolova and G.~Lapalme, ``A systematic analysis of performance measures for
  classification tasks,'' \emph{Information Processing and Management},
  vol.~45, pp. 427--437, 07 2009.

\bibitem{lee2019review}
C.~Lee and D.~J. Wilkinson, ``A review of stochastic block models and
  extensions for graph clustering,'' \emph{Applied Network Science}, vol.~4,
  no.~1, pp. 1--50, 2019.

\bibitem{kramer2008emergent}
M.~A. Kramer, E.~D. Kolaczyk, and H.~E. Kirsch, ``Emergent network topology at
  seizure onset in humans,'' \emph{Epilepsy research}, vol.~79, no. 2-3, pp.
  173--186, 2008.

\bibitem{mathur2020graph}
P.~Mathur and V.~K. Chakka, ``Graph signal processing of {EEG} signals for
  detection of epilepsy,'' in \emph{SPIN}, 2020, pp. 839--843.

\bibitem{iglesias_blind_separation}
S.~Iglesias, ``Blind separation of sparse signals diffused on graphs,'' GitHub
  repository, [Online]. Available: \url{https://github.com/iglesias/gsp\_bss},
  2018, accessed: 2024-09-26.

\end{thebibliography}


\end{document}